%% file: revsep.tex
  \documentclass{cjour}
\usepackage[T1]{fontenc}
\usepackage[latin1]{inputenc}
\usepackage{graphics}
\usepackage{longtable}
\IfFileExists{url.sty}{\usepackage{url}}
                      {\newcommand{\url}{\texttt}}

\makeatletter

\providecommand{\LyX}{L\kern-.1667em\lower.25em\hbox{Y}\kern-.125emX\@}
\let\SF@@footnote\footnote
\def\footnote{\ifx\protect\@typeset@protect
    \expandafter\SF@@footnote
  \else
    \expandafter\SF@gobble@opt
  \fi
}
\expandafter\def\csname SF@gobble@opt \endcsname{\@ifnextchar[
  \SF@gobble@twobracket
  \@gobble
}
\edef\SF@gobble@opt{\noexpand\protect
  \expandafter\noexpand\csname SF@gobble@opt \endcsname}
\def\SF@gobble@twobracket[#1]#2{}

\usepackage{amssymb}		
\usepackage{amsbsy}		
\usepackage{amstext}		
\usepackage{amsmath}		
\usepackage{latexsym}		
\usepackage{urls}		

\input{macros.tex} 

\makeatother

\begin{document}



\authorrunninghead{Michael P. Frank and M. Josephine Ammer}
\titlerunninghead{Reversible and Irreversible Space-Time Complexity Classes}






\title{Relativized Separation of Reversible and Irreversible Space-Time Complexity
Classes\thanks{
Manuscript as of May 23, 2001.  Posting to arXiv on Aug. 28, 2017.}}


\author{Michael P. Frank$^\dagger$ and M. Josephine Ammer$^\ddagger$}
\affil{$^{\dagger}$CISE Department, University of Florida; and $^{\ddagger}$EECS Department,
University of California, Berkeley}
\email{mpf@cise.ufl.edu; mjammer@eecs.berkeley.edu}

%









\dedication{Dedicated to the memory of Jon Barwise, logician and teacher}

\abstract{{\em Reversible computing} can reduce the energy dissipation
of computation, which can improve cost-efficiency in some contexts.
But the practical applicability of this method depends sensitively on
the space and time overhead required by reversible algorithms.  Time
and space complexity classes for reversible machines match
conventional ones, but we conjecture that the joint \emph{space-time}
complexity classes are different, and that a particular reduction by
Bennett minimizes the \emph{space-time product} complexity of general
reversible computations.  We provide an oracle-relativized proof of
the separation, and of a lower bound on space for linear-time
reversible simulations.  A non-oracle proof applies when a read-only
input is omitted from the space accounting.  Both constructions model
one-way function iteration, conjectured to be a problem for which
Bennett's algorithm is optimal.\\ \\
Several versions of this paper are available on the World-Wide Web at\\
\texttt{http://revcomp.info/legacy/mpf/rc/memos/M06\us oracle.html}.  }

\keywords{complexity theory, physics of computing, reversible computing, space-time complexity, incompressibility methods, relativized results, oracles, lower bounds}

\begin{article}

\contents



\raggedbottom 

\section{Introduction}
\label{s:intro}

This paper deals with \emph{reversible} models of computation, which
differ from conventional models in that all operations in a reversible
computation must be (locally) invertible.  Some discussion of the
background and motivation for such models is warranted, for the
benefit of readers who may be unfamiliar with them.

\paragraph{Importance of energy dissipation limits for future computing}
Over the history of computing, shrinking bit-device (\eg, transistor)
sizes have resulted in an energy dissipation per bit-operation that
has decreased roughly in proportion to the ever-increasing rates of
bit-operations achievable in a machine of given cost. As a result,
total power dissipation is not overwhelmingly greater today for a
machine of given cost than it was in the early days of computing, even
though the computational performance of machines has increased by many
orders of magnitude over the same period. (\Eg, Compare today's
order-100 Watt, 1 GIPS\label{p:GIPS} desktop computers with order-10
Watt, 1 IPS\label{p:IPS} hand-cranked mechanical calculators of a
hundred years ago.) So, although power requirements have always been
somewhat relevant to computer performance, they have never been the
\emph{overwhelmingly} dominant limiting factor---the total cost of the
energy needed to run a computer over its expected operational lifetime
has never been much greater than the cost of the machine itself.

However, in the future this situation could change, if continuing
improvements in manufacturing techniques, such as nano-mechanical
assembly {\cite{Drexler-92}} or molecular self-assembly\ed{[refs?]}
result in manufacturing costs per bit-device continuing to decrease
even after fundamental thermodynamic or technological lower limits on
bit-energies have been reached, and if techniques that recycle bit
energies are not applied. The Moore's law trend-line (see
fig.~\ref{f:etrend}) for bit energies reaches the absolute
thermodynamic minimum of about \( \kT \aprxeq \e {4}{-21}\units
{\Joule } \)\label{p:kB}\label{p:T}\label{p:Joule} (for
room-temperature operation) by around the year 2035, so bit energies
\emph{must} start to level off at that time, if not
earlier. Decreasing temperature would permit lower bit-energies, but
this would not by itself reduce total system power dissipation when
the cooling system is included, even if an ideal Carnot-cycle
refrigerator is used.

So, unless manufacturing costs start to level off at the same time or
earlier, or the cost of providing power and heavy-duty cooling systems
decreases rapidly, we may face the problem that although we might be
able to afford to {\emph{build}} nanocomputers with ever-increasing
numbers of bit-devices, we might not be able to {\emph{operate}} them
for very long at anywhere close to their peak performance. This
problem has been pointed out before by nanotechnology visionaries
Drexler, Merkle, Hall, and others
{\cite{Drexler-92,Merkle-93a,Hall-92-v1}}.
\begin{figure}
\resizebox*{1\textwidth}{!}{\includegraphics{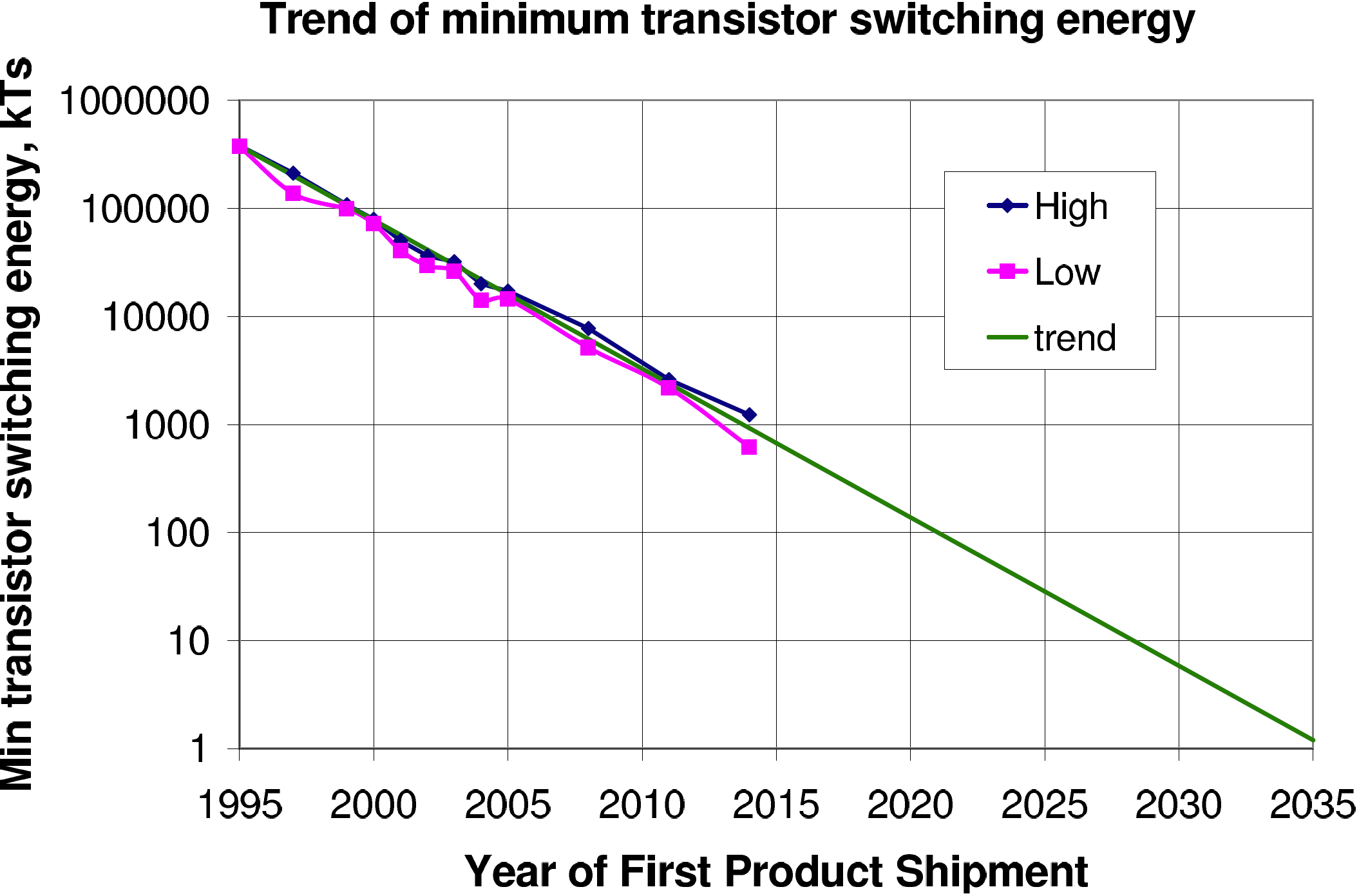}} 

\caption{\label{f:etrend}This graph shows the \protect\(
CV^{2}/2\protect \)\label{p:C}\label{p:V} energy required to charge
the gate of a minimum-sized transistor, computed from figures for
power supply voltage, minimum transistor length, and gate oxide
thickness listed in the 1994, 1997, and 1999 editions of the
International Technology Roadmap for Semiconductors (formerly the
National Semiconductor Technology Roadmap).  Values for both
high-performance and low-power design scenarios are shown. The 1999
roadmap specifies these quantities up through the year 2014. The
trendline shown extrapolates the roadmap's trends, but the roadmap has
itself historically been slightly more conservative than the
technological reality of the steady Moore's Law trends which have held
sway for more than 40 years.}
\end{figure}

Let us define the \emph{power premium} \prem\label{p:prem}\ of a
system to mean the ratio given by the expected lifetime cost of
operation of a machine's power and cooling systems, divided by the
cost of the computing hardware itself. Note that the power premium
could potentially be much greater than 1 even today in specialized
applications such as mobile computing (where there is a real, but
difficult to quantify, added cost for power in the form of
inconvenience to the user of carrying around heavy spare batteries)
and space-borne systems (where the weight of solar panels and
radiators incurs a high launch cost).  But, as various bit-energy
limits come into play in future decades, power premiums can be
expected to increase for a wider range of computing systems, if
cheaper manufacturing becomes available. Whenever \( \prem \gg 1 \),
it can make sense to change the system design in ways that incur
increased manufacturing costs in exchange for reduced power
requirements, if as the sum of the system's manufacture and operation
costs is thereby reduced.

Furthermore, under a reasonable set of physical assumptions (such as bounded
heat flow density in the cooling system) one can show \cite{Frank-etal-98,Frank-Knight-98,Frank-99b}
that for a broad class of parallel computations that require frequent intercommunication
between processing elements, asymptotically reducing energy dissipation per
operation enables strictly superior asymptotic performance even if the cost
of energy itself is negligible, since reducing heat flow enables packing devices
more densely, with shorter round-trip communication delays.

\paragraph{Adiabatic computing techniques}
For any given level of bit energies, the only way to avoid dissipating
roughly one bit-energy with each bit-operation is to use
\emph{adiabatic} (\ie, asymptotically thermodynamically reversible)
physical mechanisms to conduct the bit-operation. It is a consequence
of the second law of thermodynamics that such mechanisms are capable
of performing only \emph{logically reversible} (\ie, bijective) local
transformations of a system's digital state
\cite{Landauer-61,Leff-Rex-90,Bennett-82}. Fortunately, it turns out
that reversible operations are still computationally
universal. Several fully-reversible universal processors have already
been built \cite{Athas-etal-97b,Frank-etal-98a,Vieri-99}.

By how large a factor can adiabatic/reversible techniques reduce the
fraction of bit energies which is dissipated in practice? The precise
answer for any given bit-device technology is, as of this writing,
still unclear (although we are working on it). There are several
independent limiting factors, including the rate of energy leakage of
the bit-devices used, and the maximum efficiency (the \( Q
\)\label{p:Q} quality factor) of the energy-recovering power supplies
needed to drive adiabatic circuits.

In addition to these technological factors, there is also an important economic
limiting factor. An adiabatic reduction in the energy dissipated per operation
by a factor of \reduc\label{p:reduc} requires slowing devices down by a factor
of (at least) \reduc\ as well, so that \reduc\ times as many devices are required
to achieve a given level of raw processing performance. That is, the raw hardware
efficiency or space-time efficiency (in physical units) of an adiabatic machine
decreases in rough proportion to its increase in energy efficiency. As a result,
adiabatics cannot cost-effectively reduce power dissipation by a factor greater
than the power premium \prem, because this would raise the cost of the hardware
to be greater than the original cost of the power, thereby nullifying any economic
benefit of the decreased power consumption.

We do not yet know exactly which of these various limiting factors
will dominate in a real adiabatic computing system implementation,
because we do not yet have a sufficiently detailed adiabatic system
design, including optimized logic and power supply designs, and
accurate models of power supply dissipation and device leakage. But so
far, we know of no fundamental reasons why these technological lower
bounds on dissipation per bit-operation can not be reduced arbitrarily
through ``mere'' engineering improvements, so it seems plausible that
eventually, as power premiums increase and greater and greater degrees
of reversibility can be implemented, the hardware efficiency of larger
and larger reversible computations will come to be a dominant concern.

\paragraph{Complexity of reversible computations}
It turns out that there is an important complexity-theoretic impact on
this hardware efficiency issue. Beyond the immediate physical slowdown
by \reduc, the hardware efficiency of an adiabatic system will in
general be further decreased as a result of the possibly greater
\emph{algorithmic} space-time cost (that is, bits of state required,
times number of parallel state-update steps) for the reversible
implementation of a specific computation or sub-computation within the
machine. Many specific computations have reversible algorithms that
incur no greater space-time cost than their traditional irreversible
equivalents; some examples are mentioned in \S9.5 of
\cite{Frank-99b}. But what about other computations? In 1989, Bennett
\cite{Bennett-89} proposed a general irreversible-to-reversible
conversion technique that incurs only a modest polynomial increase in
the spacetime cost for any computation. Although originally described
as a software algorithm, it can be straightforwardly mapped to an
equivalent logic-circuit construction.

Knowing of this polynomial reduction would be enough to satisfy many
complexity theorists, but real-world concerns depend critically on
such minutae as the degree of a polynomial, or the size of its
constant coefficient. It could, for example, make the difference
between adiabatic techniques yielding significant improvements in
cost-efficiency in future generations of computing technology (or even
in near-term power-limited applications), or, in contrast, yielding no
improvements ever, depending on the absolute hardware efficiency of
substantially reversible versions of the circuit algorithms required
to implement a reasonable general-purpose microprocessor. As a result,
the outcome of a more detailed study of reversible complexity theory
is vitally relevant to planning future computing technologies.

Therefore, the question naturally arises as to whether Bennett's algorithm is
the asymptotically optimal one for conversion of arbitrary irreversible algorithms
to reversible ones, or whether a better algorithm (with still-reasonable constant
factors) might be found. If the latter were to occur, the benefits of reversible
computing might be much greater, and be realized much sooner, than would otherwise
be the case.

\paragraph{Old and new complexity conjectures}
Li and Vit\'anyi conjectured in 1996 \cite{Li-Vitanyi-96a} that
Bennett's algorithm was optimal, in terms of space complexity. Lange
{\etal}\ disproved this conjecture in \cite{Lange-etal-97}, but with a
construction that incurred exponential increases in time
complexity. However, we hypothesize that Bennett's algorithm
\emph{remains} optimal under the metric of \emph{space-time}
complexity, or space complexity \emph{multiplied by} time complexity
for a given algorithm, which is, anyway, the complexity measure that
most directly relates to the goal of maximizing hardware efficiency
(throughput per unit cost) in computer engineering.

Although our new conjecture is not yet proven, in this paper we
provide suggestive evidence in support of it, in the form of an oracle
construction that separates reversible and irreversible space-time
complexity classes, together with lower bounds which are met by
Bennett's algorithm.

\paragraph{Relevance of our relativized proof}
We are well aware that relativized constructions have no general validity in
drawing conclusions about non-relativized complexity classes, but we felt that
presenting our construction might still be useful, for several reasons:

\begin{enumerate}
\item Our oracle is designed to be as realistic as possible: Although
technically it is infeasible to physically realize exactly as defined,
it is at least computable in principle. The oracle calls are also
straightforwardly undo-able, as would be any real primitive operation
in a reversible machine.
\item The structures of the oracle, and of the language that separates
the classes, are designed to model a realistic type of real-world
computation: Namely, the iteration of an arbitrary one-way function,
such as a cryptographic hash function.  We conjecture that if one-way
functions do exist, then such iteration is a non-relativized example
of a computational problem for which a spacetime-optimal reversible
algorithm indeed results from Bennett's construction, and therefore
our lower bounds still hold without the oracle. It is conceivable that
some of the ideas or techniques used in our proof could be applied to
this one-way-function iteration scenario, to prove a separation of the
reversible and irreversible classes without resorting to an oracle,
although we have not yet seen how to do so. But, perhaps someone with
more familiarity with the theory of one-way functions would see the
trick. Therefore, we thought it worthwhile to at least present this
result to the community.
\item Finally, we feel that this entire field, which we call
``physical computing theory,'' of working with new theoretical models
of computation that are informed by increasingly-important physical
constraints such as the energy cost of bit erasure, is deserving of
more attention and we wish to help raise its visibility within the
computer science community. The increasing need for models of
computation that relate more closely to physics, and some proposed
examples of such models, are discussed in more detail in
\cite{Frank-Knight-98} and in chapters 2, 5, and 6 of
\cite{Frank-99b}.
\end{enumerate}
The results of this paper were first derived by the authors Frank and
Ammer in 1997 at MIT, and were circulated in preprint form within the
reversible computing community at that time.

\section{Table of symbols}

The following table gives the meanings of most symbols used in this document.
The third column gives the page number of the first appearance (often the definition)
of the given symbol in the text. Please note that a symbol that has different
meanings in different contexts within this paper correspondingly has multiple
entries in this table.

See also table~\ref{t:oog} on page~\pageref{t:oog} for the definitions of
our order-of-growth notations.

{\centering \begin{longtable}{cp{3.8in}c}
\hline 
Sym.&
Meaning&
p.\\
\hline 
\endhead
\emph{A}&
A particular self-reversible oracle that separates two given corresponding
\TISP\ and \RTISP\ complexity classes. Modeled as a function \( A:\tapespace \rightarrow \tapespace  \),
where \( A=A^{-1} \).&
\pageref{p:Aoracle}\\
\emph{a}&
In \S\ref{s:nonrel}, a bit-string of length \emph{b} used as an address  to
reference the memory \emph{I}.&
\pageref{p:address}\\
\emph{B}&
A particular permutation oracle that equates two given corresponding \TISP\
and \RTISP\ complexity classes.&
\pageref{p:Boracle}\\
\emph{b}&
Some arbitrary bit-string.&
\pageref{p:bitstring}\\
\emph{b}&
In \S\ref{s:nonrel}, a word length \( b\geq 0 \). Also, for \( n=b2^{b} \),
\( b(n)\equiv b \). &
\pageref{p:b-wordlen}\\
\( C \)&
Transistor gate capacitance.&
\pageref{p:C}\\
\( C_{\tau } \)&
Machine configuration of machine \( M_{i} \) resulting after \( \tau  \) steps
of execution on input \( \zerobit ^{n} \).&
\pageref{p:Ctau}\\
c&
Centi-, \( 10^{-2} \).&
\pageref{p:centi}\\
\emph{c}&
In \S\ref{s:nonrel}, a presumed constant such that reversible machine \emph{M}
decides L in no more than \( c+c\SP  \)  space and \( c+c\TI  \) time.&
\pageref{p:cnonrel}\\
\( c_{i} \)&
The constant \( c_{i}\in \N  \) appearing in the \emph{i}th pair \( (M_{i},c_{i}) \)
in an enumeration of all pairs of reversible oracle-querying machines \& such
constants.&
\pageref{p:enum}\\
\( \tapespace  \)&
The space of possible oracle tape contents.&
\pageref{p:tapespace}\\
\( \klass \) &
A variable standing for an arbitrary complexity class.&
\pageref{p:klass}\\
\emph{D}&
For a given time point \( \tau  \), a direction ({}``forwards{}'' or {}``backwards{}'')
 in which queries lie that cause most of the nodes pebbled at \( \tau  \) to be
pebbled.&
\pageref{p:D}\\
\emph{d}&
A description; a bit string that describes another bit string under some description
system \descsys.&
\pageref{p:desc}\\
\( \epsilon  \)&
An arbitrarily small positive real number; \( \epsilon \in \field {R} \)\label{p:R};
\( \epsilon >0 \); \( \epsilon \rightarrow 0 \).&
\pageref{p:epsilon}\\
\reduc&
Factor reduction in energy dissipation per operation from adiabatics.&
\pageref{p:reduc}\\
\emph{F}&
A set of functions \emph{f} having a particular asymptotic relation (\Exactly,
\Atmost, \Atleast, \Morethan, \Lessthan)&
\\
\emph{f}, \emph{g}&
In localized contexts, these are often complexity functions  \( f,g:\N \rightarrow \N  \)
mapping input lengths (in bits) to some quantity that is roughly proportional
to a complexity measure (\eg,  space or time) for worst-case inputs of the given
length.&
\pageref{t:oog}\\
\emph{f}&
In our main proof, \emph{f} is a partial \emph{successor function} \( \funcdr {f}{\bitset ^{*}}{\bitset ^{*}} \)
defining a directed graph on bit-strings that is represented by our \emph{graph
oracle} \emph{A.}&
\pageref{p:f-succfunc}\\
\gram&
Gram; unit of mass originally defined as the mass of \( 1\units {\cm ^{3}} \)\label{p:centi}
of water.&
\pageref{p:gram}\\
\emph{h}&
From given time point \( \tau  \), how many nodes are pebbled because of queries
in direction \emph{D}?&
\pageref{p:h}\\
\emph{I}&
A random-access, reversible, read-only memory of \( 2^{b} \) \emph{b}-bit words.&
\pageref{p:IROM}\\
IPS&
One instruction per second; measure of performance.&
\\
\emph{i}&
Except in localized contexts, \emph{i} in this paper means the index of one
of the possible pairs \( (M_{i},c_{i}) \) of reversible machines \& constants.&
\pageref{p:enum}\\
\emph{i}&
In \S\ref{s:nonrel}, a node index, \( 1\leq i\leq t \).&
\pageref{p:i-nodeindex}\\
\Joule&
Joule; the SI unit of energy, defined as  \( 1\units {\Newton \cdot \meter } \)\label{p:Newton}\label{p:meter}.&
\pageref{p:Joule}\\
\emph{j}&
Index of a query string, \( 1\leq j\leq t \).&
\pageref{p:j}\\
k&
Kilo-, \( 10^{3} \).&
\pageref{p:kilo}\\
\emph{k}&
Number of sublevel repetitions in Bennett's 1989 algorithm \cite{Bennett-89}.&
\pageref{p:k}\\
\emph{k}&
Index of a query string, \( 1\leq k\leq t \).&
\pageref{p:k-qindex}\\
\emph{k}&
Largest number of pebbles which is insufficient to pebble \( 2^{k} \) nodes
in Bennett's pebble game.&
\pageref{p:kpebbles}\\
\kB&
Boltzmann's constant, \( \about 1.4\times 10^{-23}\units {\Joule /\Kelvin } \).&
\pageref{p:kB}\\
\emph{L}&
For given \SP, \TI, and \emph{A,} the separator language \( L(A) \) shows \( \RST {\SP }{\TI }^{A}\nsupseteq \ST {\SP }{\TI }^{A} \); it
belongs to the latter class but not the former.&
\pageref{p:L}\\
\emph{L}&
{\raggedright In \S\ref{s:nonrel}, for given \SP, \TI, this is  the (non-relativized) language
showing that \( \RST {\SP }{\TI }\nsupseteq \ST {\SP }{\TI } \).}&
\pageref{p:Lnonrel}\\
\( \ell  \)&
A natural number giving the length of a bit-string.&
\pageref{p:ell}\\
\emph{M}&
In \S\ref{s:nonrel}, this is an (oracle-less) reversible machine presumed to
decide the language \emph{L} within \( c+c\SP  \) space and \( c+c\TI  \)
time.&
\pageref{p:Mnonrel}\\
\( M_{i} \)&
The reversible oracle-querying machine in the \emph{i}th pair \( (M_{i},c_{i}) \)
of an enumeration of all pairs of such machines and constant factors.&
\pageref{p:enum}\\
\meter&
Meter; unit of length originally defined as \( \frac{1}{4}\times 10^{-7} \)
of Earth's circumference.&
\pageref{p:meter}\\
\Newton&
Newton; the SI unit of force, defined as  \( 1\units {\meter \cdot \kg /\second ^{2}} \)\label{p:kilo}\label{p:gram}\label{p:second}.&
\pageref{p:Newton}\\
\N&
The set of the natural numbers, \{0, 1, 2, ...\}.&
\pageref{p:N}\\
\NP&
The complexity class of languages decidable in polynomial time by nondeterministic
Turing machines.&
\pageref{p:NPclass}\\
\emph{n}&
Number of levels in Bennett's 1989 algorithm \cite{Bennett-89}.&
\pageref{p:nBen}\\
\emph{n}&
Abbreviation of \nin\ or \( n_{i} \).&
\pageref{p:n}\\
\( n_{i} \)&
The length of input strings for which machine \( M_{i} \) fails to decide \emph{L}
within the space-time bounds determined by \SP, \TI, and constant \( c_{i} \).&
\pageref{p:ni}\\
\nin&
Number of bits in an input string.&
\pageref{p:nin}\\
\textsc{next}&
Given time point \( \tau  \), \textsc{next}\( (q_{j} \)) is the next query in \( M_{i} \)'s
history involving \( q_{j} \) before time \( \tau  \).&
\pageref{p:next}\\
\emph{O}&
Some arbitrary oracle. (In our context, a self-reversible one.)&
\pageref{p:O-oracle}\\
\Atmost&
The {}``at most{}'' order-of-growth operator \( \Atmost  \) maps any function
\( g:\N \rightarrow \N  \) to the set \( F=\Atmost (g) \) of functions \emph{f}
that are asymptotically at most proportional to \emph{g}.&
\pageref{t:oog}\\
\Lessthan&
The {}``less than{}'' order-of-growth operator \Lessthan\ maps any function
\( g:\N \rightarrow \N  \) to the set \( F=\Lessthan (g) \) of functions \emph{f}
that are asymptotically strictly less than \emph{g}.&
\pageref{t:oog}\\
\P&
The complexity class of languages decidable in polynomial time in most traditional
models of computation (\eg, Turing machines).&
\pageref{p:Pclass}\\
\prem&
Power premium; ratio of lifetime power cost to logic hardware cost.&
\pageref{p:prem}\\
\emph{p}&
The number of nodes that are pebbled at time \( \tau  \).&
\pageref{p:p}\\
\textsc{prev}&
Given time point \( \tau  \), \textsc{prev}\( (q_{j} \)) is the previous query in \( M_{i} \)'s
history involving \( q_{j} \) before time \( \tau  \).&
\pageref{p:prev}\\
\emph{Q}&
Quality factor; ratio between energy transfered and energy dissipated during
a system's cycle of operation.&
\pageref{p:Q}\\
\emph{q}&
A possible oracle query string, \ie, an oracle tape contents, \ie, a bit string,
\ie, a graph node identifier, \ie, a graph node.&
\pageref{p:query}\\
\( q' \)&
An alternative final node in the chain, replacing our original choice of \( q_{t} \).&
\pageref{p:qprime}\\
\( q_{0} \)&
Initial query string in a node chain. \( q_{0}={\texttt {0}}^{\SP } \).&
\pageref{p:q0}\\
\( q_{j} \)&
A particular query string in the sequence \( q_{1},\ldots ,q_{t} \) formed
from \emph{x,} or if \( j=0 \), see \( q_{0} \) above.&
\pageref{p:qj}\\
\( \field {R} \)&
The (nonconstructive) set of all {}``real{}'' numbers.&
\pageref{p:R}\\
\RTISP&
\RST{\SP}{\TI} is the complexity class of problems solvable by reversible algorithms
taking time \( \Atmost (\TI ) \) and space \( \Atmost (\SP ) \).&
\pageref{p:RST}\\
\( r(I) \)&
In \S\ref{s:nonrel}, the 1-bit result for a given input memory \emph{I}, found
by doing \( \lfloor \TI /\SP \rfloor  \) iterated pointer dereferences in \emph{I}
starting at address \( {\texttt {0}}^{b} \).&
\pageref{p:result}\\
\SP&
Space bounding function \( \SP :\N \rightarrow \N  \), mapping an input  length \nin
to an upper bound \( \SP (\nin ) \) on the number of temporary state bits used
at any time in processing any input of length \nin.&
\pageref{p:SP}\\
\( \SP'  \)&
A larger space bounding function, \( \SP '\lessthan \SP \log (\TI /\SP ) \), which
is \emph{still} not enough to allow reversible machines to compute the same functions
in linear time (in our oracle model).&
\pageref{p:SPprime}\\
\second&
Second; unit of time originally defined as 1/86,400 of Earth's solar day.&
\pageref{p:second}\\
\descsys&
A description system. (In \S\ref{s:nonrel}, a particular one that we are defining.)&
\pageref{p:descsys}\\
\( \descsys _{i} \)&
The particular description system used to select the incompressible string \emph{x}
that defines the chain of nodes that foils \( M_{i} \).&
\pageref{p:descsysi}\\
\temperature&
Absolute temperature.&
\pageref{p:T}\\
\TI&
Time bounding function \( \TI :\N \rightarrow \N  \), mapping an input length
\nin to an upper bound \( \TI (\nin ) \) on the number of state-update {}``ticks{}''
to be used in processing any input of length \nin.&
\pageref{p:TI}\\
\( \TI ' \)&
Actual number of steps \( \TI '\leq c_{i}+c_{i}\TI (n) \) taken before halting in
the case of a machine that does not exceed the time bound.&
\pageref{p:TIprime}\\
\TISP&
\ST{\SP}{\TI} is the complexity class of problems solvable by ordinary algorithms
taking time \( \Atmost (\TI ) \) and space \( \Atmost (\SP ) \).&
\pageref{p:ST}\\
\emph{t}&
\( t(n) \) is the number of nodes, size \( \SP (n) \) each, in a chain  of
nodes that will take time \( \Exactly (\TI (n)) \) to  traverse on a serial
machine. \( t(n)\equiv \lfloor \TI (n)/\SP (n)\rfloor  \).&
\pageref{p:t}\\
\( \tau  \)&
\( 0\leq \tau \leq \TI ' \), an index of the machine configuration of \( M_{i} \) (running
on the oracle graph) that results after \( \tau  \) steps (primitive operations)
have taken place.&
\pageref{p:tau}\\
\( \Delta \tau _{j} \)&
The number of steps between time \( \tau  \) and the query in direction \emph{D} that
causes node \( q_{j} \) to be pebbled at time \( \tau  \).&
\pageref{p:Deltatauj}\\
\Exactly&
The {}``exactly{}'' order-of-growth operator \Exactly\ maps any function \( g:\N \rightarrow \N  \)
to the equivalence class \( F=\Exactly (g) \) of functions \emph{f} that are
asymptotically proportional to \emph{g}.&
\pageref{t:oog}\\
\( V \)&
Logic swing voltage; absolute voltage difference between \zerobit\ and \onebit\
logic levels.&
\pageref{p:V}\\
\emph{w}&
An arbitrary bit-string input to our oracle-querying machines.&
\pageref{p:w}\\
\( w_{i} \)&
In \S\ref{s:nonrel}, length-\emph{b} bit string number \emph{i}, where \( 1\leq i\leq t \), in
a linked list of bit strings formed from \emph{x}.&
\pageref{p:wi-nonrel}\\
\emph{x}&
A bit-string, \( |x|=\TI (n) \), incompressible in description system \( s_{i} \), to
be broken up into a chain of node bit-strings \( q_{0},\ldots ,q_{t} \).&
\pageref{p:x}\\
\( x' \)&
\emph{x} with a substring spliced out. (See explanations in text.)&
\pageref{p:xprime}\\
\emph{y}&
A bit-string that is described by another bit-string \emph{d} under some description
system \emph{s}.&
\pageref{p:describee}\\
\emph{z}&
The maximum length over all oracle queries asked by machines  \( M_{0},\ldots ,M_{i-1} \)
running within their respective bounds \( c_{0},\ldots ,c_{i-1} \) when given
respective inputs \( \zerobit ^{n_{0}},\ldots ,\zerobit ^{n_{i-1}} \).&
\pageref{p:z}\\
\Atleast&
The {}``at least{}'' order-of-growth operator \Atleast\ maps any function
\( g:\N \rightarrow \N  \) to the set \( F=\Atleast (g) \) of functions \emph{f}
that are asymptotically no less than proportional to \emph{g}.&
\pageref{t:oog}\\
\Morethan&
The {}``more than{}'' order-of-growth operator \Morethan\ maps any function
\( g:\N \rightarrow \N  \) to the set \( F=\Morethan (g) \) of functions \emph{f}
that are asymptotically strictly greater than \emph{g}.&
\pageref{t:oog}\\
\hline 
\end{longtable}\par}

\section{Review of previous results in reversible computing theory}
\label{s:review}

In this section we briefly review previous results in the theory of computability
and of computational complexity relating to reversible computation.

\paragraph{Reversible models of computation}
Reversible models of computation can be easily defined in general as models
of computation in which the transition function between machine configurations
has a single-valued inverse. In other words, the directed graph showing allowed
transitions between states has in-degree 1. In this paper we will always deal
with machines that are deterministic, so that the configuration graph always
has out-degree one as well. See figure~\ref{f:revers}.
\begin{figure}
\centerline{\resizebox*{0.75\textwidth}{!}{\includegraphics{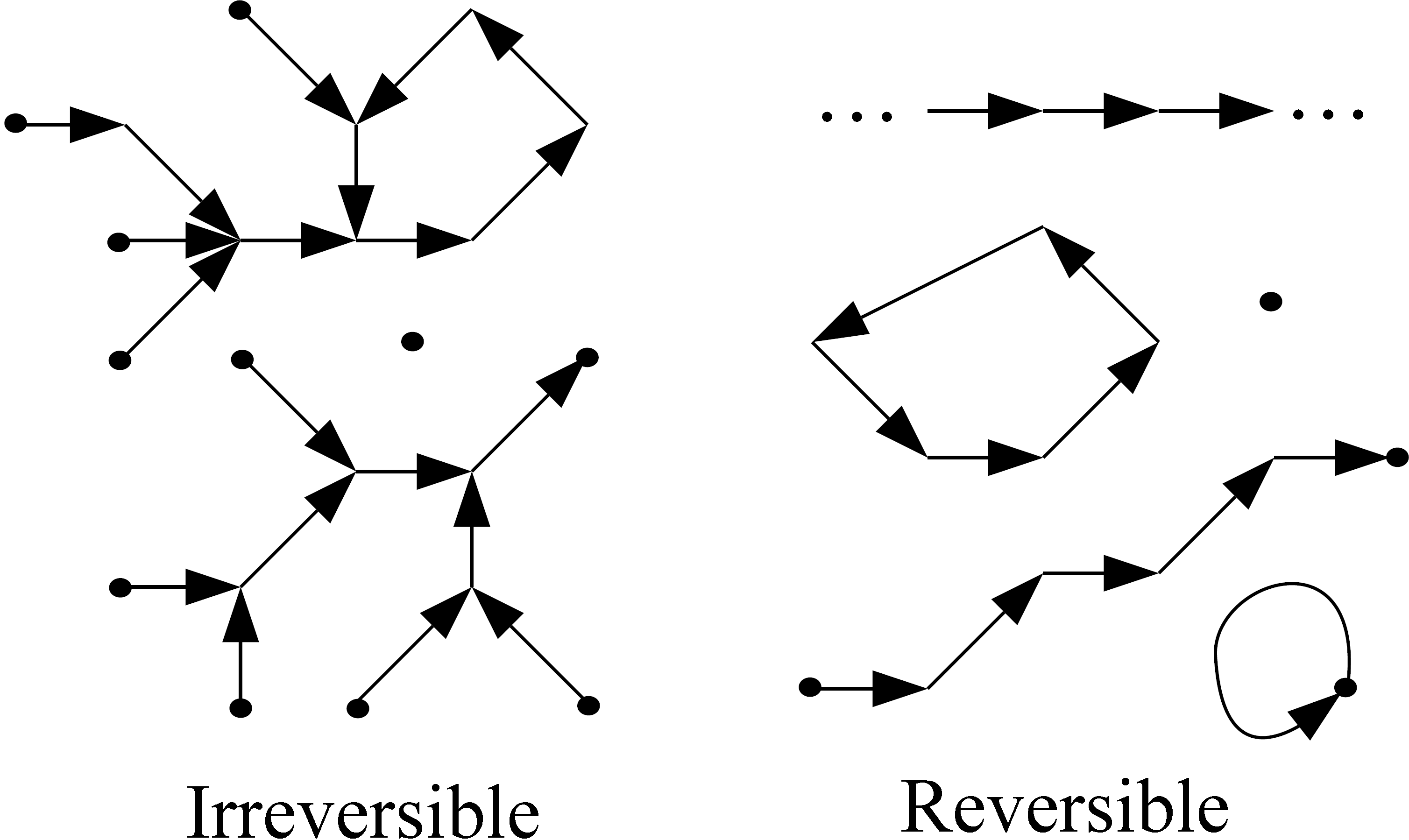}}}

\caption{Machine configuration graphs in (deterministic) reversible and irreversible
models of computation.\label{f:revers}
\captionpars

In the configuration graphs of irreversible machines, configurations may have
many different predecessor configurations. In reversible models of computation,
each configuration may have at most one predecessor. The configuration graph
therefore consists of disjoint loops and chains, which may be infinite. In both
reversible and irreversible models we may, if we wish, permit configurations
having 0 predecessors (initial states) and/or 0 successors (final states).
}

\end{figure}

\subsection{Computability in reversible models}

\paragraph{Unbounded-space reversible machines are Turing-universal}
In his 1961 paper \cite{Landauer-61}, Landauer had already pointed out that
arbitrary irreversible computations could be embedded into reversible ones by
simply saving a record of all the information that would otherwise be thrown
away (\cf~\S{}3 of~\cite{Landauer-61}). This observation makes it obvious
that reversible machines with unbounded memory can certainly compute all the
Turing-computable functions.

We call this idea, of embedding an irreversible computation into a reversible
one by saving a history of garbage, a {}``Landauer embedding,{}'' since Landauer
seems to have been the first to suggest it.

\paragraph{A certain model of reversible finite automata is especially weak}
In contrast to Landauer's result, in 1987 Pin \cite{Pin87} investigated reversible
\emph{finite} automata, which he defined as machines with fixed memory reading
an unbounded-length \emph{one-way} stream of data, and found that such cannot
even decide all the regular languages, which means that technically they are
strictly less powerful than normal irreversible finite automata.

So there are stream recognition tasks computable by an irreversible machine
with fixed memory that no purely reversible machine with fixed memory can compute,
given an external one-way stream of input. We should note, however, that this
incapacity may be due to the non-reversible nature of the input flow, rather
than to the reversibility of the finite automaton itself. Conceivably, if a
finite reversible machine was permitted to read backwards as well as forwards
through its read-only input, and perform some sort of {}``unread{}'' operations,
it might then be able to recognize any regular language. But we have not investigated
that possibility in detail.

In any event, the finite automaton model is not generally considered
to express the salient features of computation, since real computers
are not designed as state machines with small fixed numbers of states,
but rather as unbounded-memory machines that can be given as much
external storage as needed to perform a particular task, and can
explore an enormous state space, one that grows exponentially with the
number of storage bits that are available.  So, for the rest of this
paper, we consider only models of computation that permit access to
increasingly large amounts of memory as input sizes increase.  For
such machines, Landauer's result overrides Pin's, and pure
computability is no longer an issue.  So we turn to questions of
computational complexity.

\subsection{Time complexity in reversible models}
\label{s:time-cpxy}

In a theoretical computer science context, {}``time complexity{}'' \TI\label{p:TI}
for serial machine models means essentially the number of primitive operations
performed. Landauer's suggestion (\cf~\S{}3 of~\cite{Landauer-61}) of embedding
each irreversible operation into a reversible one makes it clear that the number
of such operations in a reversible machine need not be larger than the number
for an irreversible machine, as was demonstrated more explicitly in later embeddings
by Lecerf \cite{Lecerf-63} and Bennett \cite{Bennett-73}. So under the time
complexity measure by itself, reversibility does not hurt.

Can a reversible machine perform a task using \emph{fewer}
computational operations than any irreversible machine?  Obviously
not, if we take reversible operations to just be a special case of
irreversible operations. However, it is interesting to note that,
physically speaking, actually it is the converse that is true:
so-called {}``irreversible{}'' operations, implemented physically, are
really just a special case of reversible operations, since physics is
believed to be \emph{always} reversible at a low level. The
implications of this fact for \emph{physical} time complexity are
discussed in more detail in \cite{Frank-99}. But, using the usual
computer-science definition of time as the number of
\emph{computational} operations required, clearly reversible machines
can be no more {}``time{}''-efficient than irreversible ones.

Although Lecerf and Bennett explicitly discussed their time-efficient reversible
simulations only in the context of Turing machines, the approach is easily generalized
to any model of computation in which we can give each processing element access
to an unbounded amount of auxiliary unit-access-time stack storage. For example,
Toffoli \cite{Toffoli-77} describes how one can use essentially the same trick
to create a time-efficient simulation of irreversible cellular automata on reversible
ones, by using an extra dimension in the cell array to serve as a garbage stack
for each cell of the original machine.

\subsection{Entropic complexity in reversible models}
\label{s:entcpx}

The original point of reversibility was not to reduce time but to reduce energy
dissipation, or in other words entropy production. Can this be done by reversible
machines? In 1961 Landauer \cite{Landauer-61} argued that it could not, since
if we cannot get rid of the {}``garbage{}'' bits that are accumulated in memory,
they just constitute another form on entropy, no better in the long term than
the kind produced if we just irreversibly dissipated those bits into physical
entropy right away.

\paragraph{Lecerf reversal}\label{p:lecrev}
However, in 1963, Lecerf \cite{Lecerf-63} formally described a construction
in which an irreversible machine was embedded into a reversible one that first
simulated the irreversible machine running forwards, then turned around and
simulated the irreversible machine in reverse, uncomputing all of the history
information and returning to a state corresponding to the starting state. If
anyone familiar with Landauer's work had noticed Lecerf's paper in the 1960's,
it would have seemed tantalizing, because here was Lecerf showing how to reversibly
get rid of the garbage information that was accumulated in Landauer's reversible
machine in lieu of entropy. So maybe the entropy production can be avoided after
all!

Unfortunately, Lecerf was apparently unaware of the thermodynamic implications
of reversibility; he was concerned only with determining whether certain questions
about reversible transformations were decidable. Lecerf's paper did not address
the issue of how to get useful results out of a reversible computation. In Lecerf's
embedding, by the time the reversible machine finishes its simulation of the
irreversible machine, any outputs from the computation have been uncomputed,
just like the garbage. This is not very useful!

\paragraph{The Bennett trick}
Fortunately, in 1973, Charles Bennett \cite{Bennett-73}, who was
unaware of Lecerf's work but knew of Landauer's, independently
rediscovered Lecerf reversal, and moreover added the ability to retain
useful output. The basic idea was simple: one can just reversibly copy
the desired output into available memory before performing the Lecerf
reversal! As far as we know, this simple trick had not previously
occurred to anyone.

Bennett's idea suddenly implied that reversible computers could in principle
be \emph{more} efficient than irreversible machines under at least \emph{one}
cost measure, namely entropy production. To compute an output on an irreversible
machine, one must produce an amount of entropy roughly equal to the total number
of (irreversible) operations performed; whereas the reversible machine in principle
can get by with \emph{no} new entropy production, and with the accumulation
of only the desired output in memory.

\paragraph{Entropy proportional to speed}
Unfortunately, absolutely zero entropy generation per operation is achievable
in principle only in the ideal limit of a perfectly-isolated ballistic (frictionless)
system, or in a Brownian-motion-based system that makes zero progress forwards
through the computation on average, and takes \( \Exactly (n^{2}) \) expected
time before visiting the \( n \)th computational step. In useful systems that
progress forwards at a positive constant speed, the entropy generation per operation
appears to be, at minimum, proportional to the speed. (We do not yet know how
fundamental this relationship is, but it appears to be the case empirically.)
A cost analysis that takes both speed and entropy into account will need to
recognize this tradeoff. We do this in \cite{Frank-etal-98,Frank-Knight-98}
and in ch.~6 of \cite{Frank-99}.

\subsection{Space complexity in reversible models}
\label{s:spc-cpxy}

In computational complexity, {}``space complexity{}'' refers to the number
\( \SP  \)\label{p:SP} of memory cells that are required to perform a computation.

\paragraph{Initial estimates of space complexity}
As Landauer pointed out \cite{Landauer-61}, his simple strategy of saving all
the garbage information appears to suffer from the drawback that the amount
of garbage that must be stored in digital form is as large as the amount of
entropy that would otherwise have been generated. If the computation performs
on average a constant number of irreversible bit-erasures per computational
operation, then this means that the memory usage becomes proportional to the
number of operations. This means a large asymptotic increase in memory usage
for many problems; up to exponentially large. Even if the garbage is uncomputed
using Lecerf reversal, this much space will still be needed temporarily during
the computation.

\paragraph{Bennett's pebbling algorithm}
In 1989, Bennett \cite{Bennett-89} introduced a new, more space-efficient reversible
simulation for Turing machines. This new algorithm involved doing and undoing
various-sized portions of the computation in a recursive, hierarchical fashion.
Figure~\ref{f:bennett} is a schematic illustration of this process. We call
this the {}``pebbling{}'' algorithm because the algorithm can be seen as a
solution to a sort of {}``pebble game{}'' or puzzle played on a one-dimensional
chain of nodes, as described in detail by Li and Vit\'{a}nyi '96 \cite{Li-Vitanyi-96b}.
(Compare figure~\ref{f:bennett}(a) with fig.~\ref{f:game} on page~\pageref{f:game}.)
We will discuss the pebble game interpretation and its implications in more
detail in~\S\ref{s:main}.

\begin{figure}
\centerline{\resizebox*{1\textwidth}{!}{\includegraphics{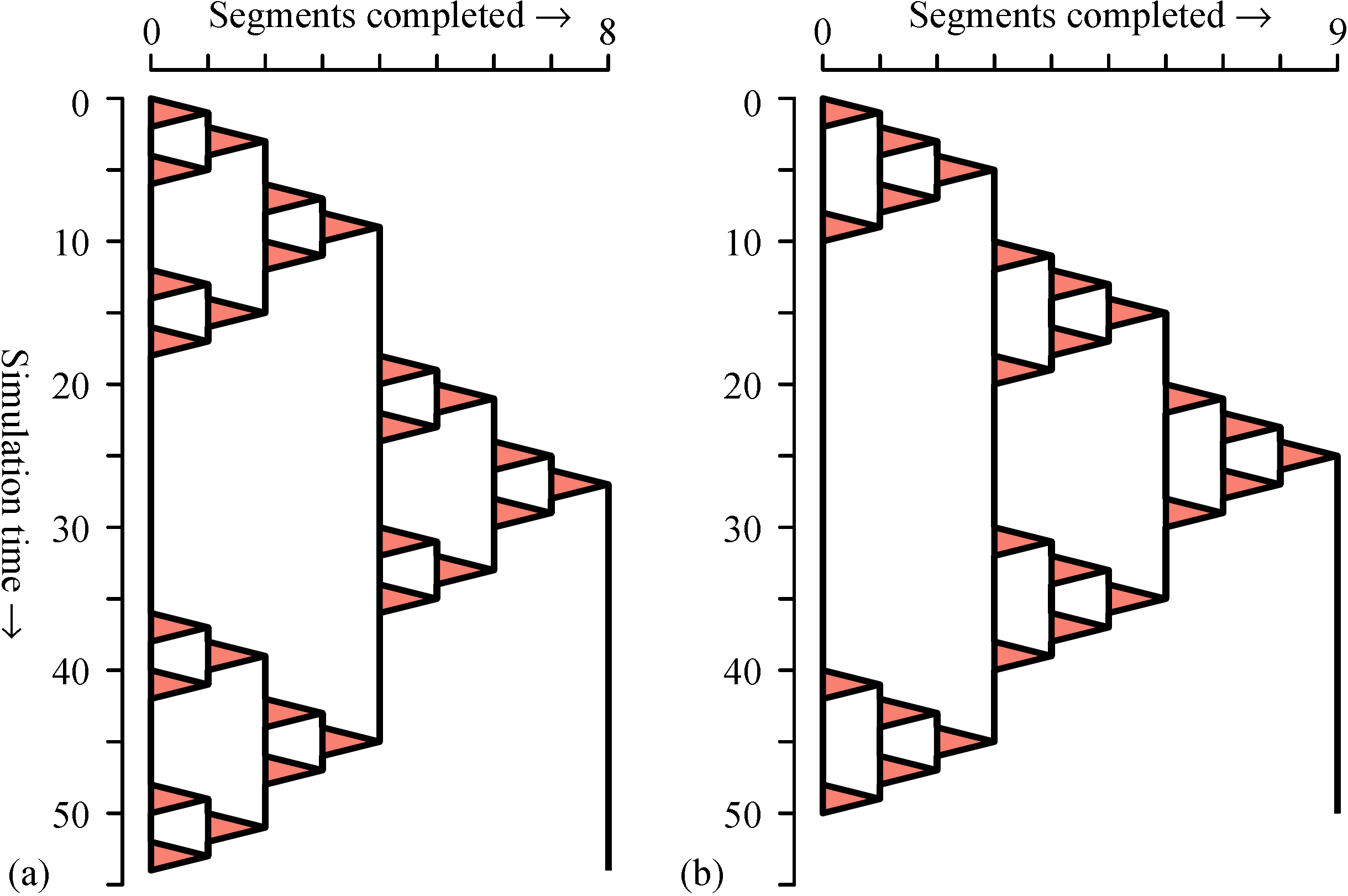}} }

\caption{\label{f:bennett}Illustration of two versions of Bennett's 1989 algorithm
for reversible simulation of irreversible machines. Diagram (a) illustrates
the version with \protect\( k=2\protect \), diagram (b) the version with \protect\( k=3\protect \).
(See text for explanation of \protect\( k\protect \).)
\captionpars

In both diagrams, the horizontal axis indicates which segment of the original
irreversible computation is being simulated, whereas the vertical axis tracks
time taken by the simulation in terms of the time required to simulate one segment.
The black vertical lines represent times during which memory is occupied by
an image of the irreversible machine state at the indicated stage of the irreversible
computation, whereas the shaded areas within the triangles represent memory
occupied by the storage of garbage data for a particular segment of the irreversible
computation being simulated.

Note that in (b), where \( k=3 \), the 9th stage is reached after only
25 time units, whereas in (a) 27 time units are required to only reach
stage 8. But note also that in (b), at time 25, five checkpoints
(after the initial state) are stored simultaneously, whereas in (a) at
most four are stored at any given time. This illustrates the general
point that higher-\( k \) versions of the Bennett algorithm run
faster, but require more memory. }

\end{figure}

The overall operation of the algorithm is as follows. The irreversible
computation to be simulated is broken into fixed-size segments, whose
run time is proportional to the memory required by the irreversible
machine. The first segment is reversibly simulated using a Landauer
embedding (as in \cite{Landauer-61}). Then the state of the
irreversible machine being simulated is checkpointed using the Bennett
trick of reversibly copying it to free memory. Then, we do a Lecerf
reversal (\S\ref{s:entcpx}, p.\ \pageref{p:lecrev}) to clean up the
garbage from simulating the first segment.

We proceed the same way through the second segment, starting from the first
checkpoint, to produce another checkpoint. After some number \( k \)\label{p:k}
of repetitions of this procedure, all the previous checkpoints are then removed
by reversing everything done so far except the production of the final checkpoint.
Now we have only a single checkpoint which is \( k \) segments along in the
computation. We repeat the above procedure to create another checkpoint located
another \( k \) segments farther along, and then again, and again \( k \)
times, then reverse everything again at the higher level to proceed to a point
where we only have checkpoint number \( k^{2} \) in memory. The procedure can
be applied indefinitely at higher and higher levels.

In general, for any number \( n \)\label{p:nBen} of recursive higher-level
applications of this procedure, \( k^{n} \) segments of irreversible computation
are be simulated by \( (2k-1)^{n} \) reversible forwards-and-backwards simulations
of a single segment, while having at most \( n(k-1) \) intermediate checkpoints
in memory at any given time \cite{Bennett-89}.

The upshot is that if the original irreversible computation takes time \( \TI  \)
and space \( \SP  \), then the reversible simulation via this algorithm takes
time \( O(\TI ^{1+\epsilon }) \)\label{p:epsilon} and space \( O(\SP \log \TI )=O(\SP ^{2}) \).
As \( k \) increases, the \( \epsilon  \) approaches 0 (very gradually), but
unfortunately the constant factor in the space usage increases at the same time
\cite{LevSher90}.

Li and Vit\'{a}nyi '96 \cite{Li-Vitanyi-96b} proved that Bennett's algorithm
(with \( k=2 \)) is the most space-efficient possible pebble-game strategy
for reversible simulation of irreversible machines. This result is central to
our proof.

Crescenzi and Papadimitriou '95 \cite{CrePap95} later extended Bennett's technique
to provide space-efficient reversible simulation of \emph{nondeterministic}
Turing machines as well.

\subsubsection{Achieving linear space complexity}
\label{s:linspace}

Bennett's results stood for almost a decade as the most
space-efficient reversible simulation technique known, but in 1997,
Lange, McKenzie, and Tapp \cite{Lange-etal-97} showed how to simulate
Turing machines reversibly in linear space---but using worst-case
exponential time. Their technique is very clever, but simple in
concept: Given a configuration of an irreversible machine, they show
that one can reversibly enumerate its possible predecessors. Given
this, starting with the initial state of the irreversible machine, the
reversible machine can traverse the edges of the irreversible
machine's tree of possible configurations in a reversible {}``Euler
tour.{}'' (See figure~\ref{f:lmt}.) This is analogous to using the
{}``right-hand rule{}'' technique (move forward while keeping your
right hand on the wall) to find the exit of a planar non-cyclical
maze. The search for the final state is kept finite, and the space
usage is kept small, by cutting off exploration whenever the
configuration size exceeds some limit. Unfortunately, the size of the
pruned tree, and thus the time required for the search, is still, in
the worst case, exponential in the space bound.

In a very recent result, Buhrman \etal\ '01 \cite{Buhrman-etal-01}
show that it is actually possible to view the Bennett and
Lange-McKenzie-Tapp techniques as extreme points on a continuous
spectrum of simulation algorithms having intermediate asymptotic space
and time requirements. (Unfortunately, all of the intermediate
algorithms in this tradeoff space still suffer at least a polynomial
increase in spacetime complexity.)

\begin{figure}
\centerline{\resizebox*{4in}{!}{\includegraphics{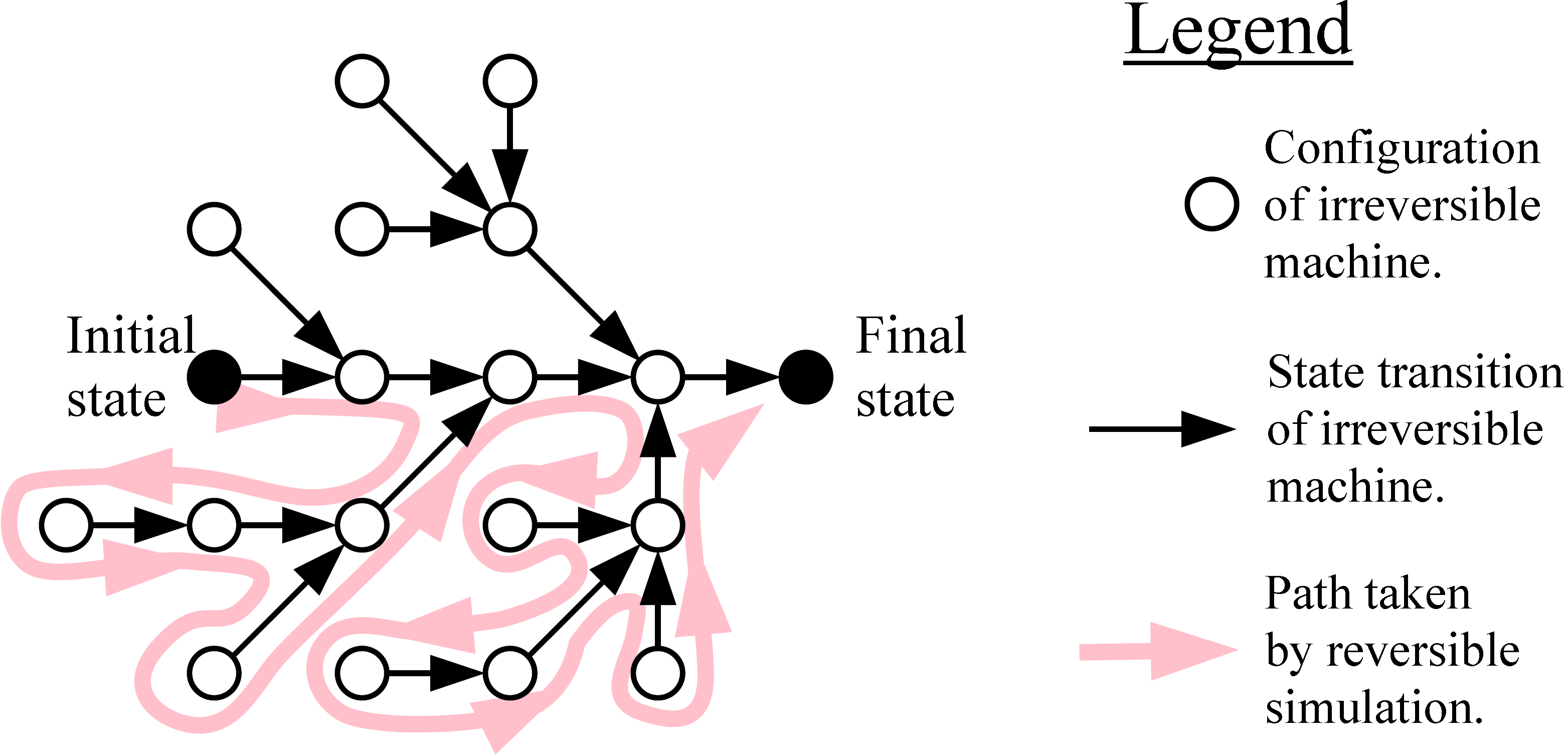}} }

\caption{\label{f:lmt}Illustration of an Euler tour of an irreversible machine's computation
tree. Although the tree has branches, the Euler tour is itself both forward-
and reverse-deterministic, and so can be traversed in purely reversible fashion,
using no more space than is needed to keep track of the current irreversible
machine configuration \cite{Lange-etal-97}.}

\end{figure}

As with Bennett's techniques, the Lange-McKenzie-Tapp technique was defined
explicitly only in terms of Turing machines, but it is easily generalized to
many different models of computation.

\noheading The above time and space complexity results for reversible
simulation are very interesting in themselves, but to our knowledge,
before our work no one had previously addressed the specific question
of whether a single reversible simulation could run both in linear
time like Bennett's 1973 technique \emph{and} in linear space like the
newer Lange \etal\ technique. Li and Vit\'{a}nyi's analysis
\cite{Li-Vitanyi-96b} of Bennett's 1989 algorithm \cite{Bennett-89}
leads to our proof in this paper that if such an ideal simulation
exists, it would not relativize to oracles, or work in cases where the
space bound is much less than the input length.

\subsection{Miscellaneous developments}

Here, we mention in passing a couple of other miscellaneous developments in
reversible computing theory.

Coppersmith and Grossman (1975, \cite{CG75}) proved a result in group theory
which implies that reversible boolean circuits only 1 bit wider than a fixed-length
input can compute arbitrary boolean functions of that input.

Toffoli (1977, \cite{Toffoli-77}) showed that reversible cellular automata
can simulate irreversible ones in linear time using an extra spatial dimension.
Fredkin and Toffoli developed much reversible boolean-circuit theory (1980--1982,
\cite{Toffoli80,Toffoli-80a,Fredkin-Toffoli-82}).

\section{General definitions}
\label{s:gendef}

In this section we set forth some general definitions that we will use
in our proof, but that may also be useful for future proofs in
reversible computing theory. Later, in section~\ref{s:specdef}, we
will give some additional, more specific definitions that are not
anticipated to be widely useful outside of this paper.

\subsection{Space-time complexity classes}
\label{s:stclasses}

Given any rever\-sible model of computation (\eg, reversible Turing
machines), and given any computational space and time bounding
functions \( \SP (\nin ),\TI (\nin ) \), we define (following Bennett
{\cite{Bennett-89}} the \defn{reversible space-time\/ \(\SP,\TI\)
complexity class}, abbreviated \( \RST {\SP }{\TI } \)\label{p:RST},
to be the set of languages that are accepted by reversible machines
that take worst-case space of \( \Atmost (\SP (\nin )) \) memory bits
and worst-case time \( \Atmost (\TI (\nin )) \) ticks, where \( \nin
\) is the length of the input. Similarly, we define the (unrestricted)
\defn{space-time\/ \(\SP,\TI\) complexity class}, abbreviated \( \ST
{\SP }{\TI } \)\label{p:ST}, to be the set of languages accepted in
that same order of space and time on the corresponding normal machine
model, without the restriction on the in-degree of the transition
graph. For oracle-rel\-a\-tiv\-ized complexity classes, we use the
notation \( \klass ^{O} \)\label{p:klass}\label{p:O-oracle}, as is
standard in complexity theory, to indicate the class of problems that
can be solved by the machines that define the class \( \klass \) if
they are allowed to query oracle \( O \).

We want to know whether \( \RST {\SP }{\TI }\qeq \ST {\SP }{\TI } \), for all
\( \SP ,\TI  \), in normal sorts of serial computational models such as multi-tape
Turing machines or RAM machines.

Unfortunately, we have found this question, in its purest form, very difficult
to definitively resolve. We do not see any general way to simulate normal machines
on reversible machines without suffering asymptotic increases in either the
time or space required. But neither do we know of a language that can be proven
to require extra space or time to recognize reversibly in ordinary machine models.
The difficulty is in constructing a proof that rules out all reversible algorithms,
no matter how subtle or clever.

But is the \( \RTISP \qeq \TISP \) question truly difficult to
resolve, or have we just been unlucky in our search for a proof? Often
in computational complexity theory, we find ourselves unable to prove
whether or not two complexity classes (for example, \P\label{p:Pclass}
and \NP\label{p:NPclass}) are equivalent. Traditionally (as in Baker
\etal\ \cite{BGS}), one way to indicate that such an equivalence might
really be difficult to prove is to show that if the machine model
defining each class is augmented with the ability to perform a new
type of operation (a query to a so-called {}``oracle{}''), then the
classes may be proven either equal or unequal, depending on the
behavior of the particular oracle. This shows that any proof equating
or separating the two classes must make use of the fact that normal
machine models are only capable of performing a particular limited set
of primitive operations. Otherwise, we could just add the appropriate
oracle call as a new primitive operation, and invalidate the supposed
proof.

In this section we will demonstrate, for any given \( \SP ,\TI  \) in a large
class, an oracle \( A \)\label{p:Aoracle} relative to which we prove \( \RST {\SP }{\TI }^{A}\neq \ST {\SP }{\TI }^{A} \),
for the case of serial machine models with a certain kind of oracle interface.
For these same \( \SP ,\TI  \) we have not yet found an alternative oracle
\( B \)\label{p:Boracle} for which \( \RST {\SP }{\TI }^{B}=\ST {\SP }{\TI }^{B} \),
except for irreversible oracles which make the equivalence trivial. It may be
that no reversible oracle that equates the classes exists, but this is uncertain.

\subsection{Reversible oracle interface}
\label{s:revoif}

First, we define an oracle interface that allows a reversible machine to call
an oracle. Ordinarily, oracle queries are irreversible, and thus impossible
in reversible machines. For example, a bit of the oracle's answer cannot just
overwrite some storage location, because regardless of whether the location
contained \zerobit\label{p:zerobit} or \onebit\label{p:onebit} before the
oracle call, after the call it would contain the oracle's answer. The resulting
configuration would thus have two predecessors, and the machine would be irreversible.

Our reversible oracle-calling protocol is as follows. Machines will have reversible
read and write access to a special \defn{oracle tape} which has a definite start,
unbounded length, and is initially clear. At any time, the machine is allowed
to perform an \defn{oracle call}, a special primitive operation which in a single
step replaces the entire contents of the oracle tape with new contents, according
to some fixed invertible mapping \( A:\tapespace \rightarrow \tapespace  \)\label{p:tapespace}
over the space \( \tapespace  \) of possible tape contents. The function \( A \)
is called a \defn{permutation oracle}. Further, if \( A \) is its own inverse,
\( A=A^{-1} \), it will be called \defn{self-reversible}. Presented more formally:
\begin{definition}
A \defn{permutation oracle} \( A \) is an invertible (bijective) function \mbox{\( A:\tapespace \rightarrow \tapespace  \)},
where \( \tapespace  \) is the space of possible contents of a semi-infinite
\defn{oracle tape}.
\end{definition}
\begin{definition}
A \defn{self-reversible (permutation) oracle} is a permutation oracle \( A \)
such that \( A=A^{-1} \).
\end{definition}

In the below, we will deal only with self-reversible
oracles. Self-reversibility ensures that machines can easily undo
oracle operations, just as they can easily undo their own internal
reversible primitives. (Since primitive operations such as
bit-operations are by definition finite operations over small
state-spaces, if those operations are invertible then their inverses
must be easy to compute.)  \ed{Give examples?} If oracle calls were
much harder to undo than to do, then the oracle model would be
unlikely to teach us anything meaningful about real machines.

\begin{figure}
\centerline{\resizebox*{0.8\textwidth}{!}{\includegraphics{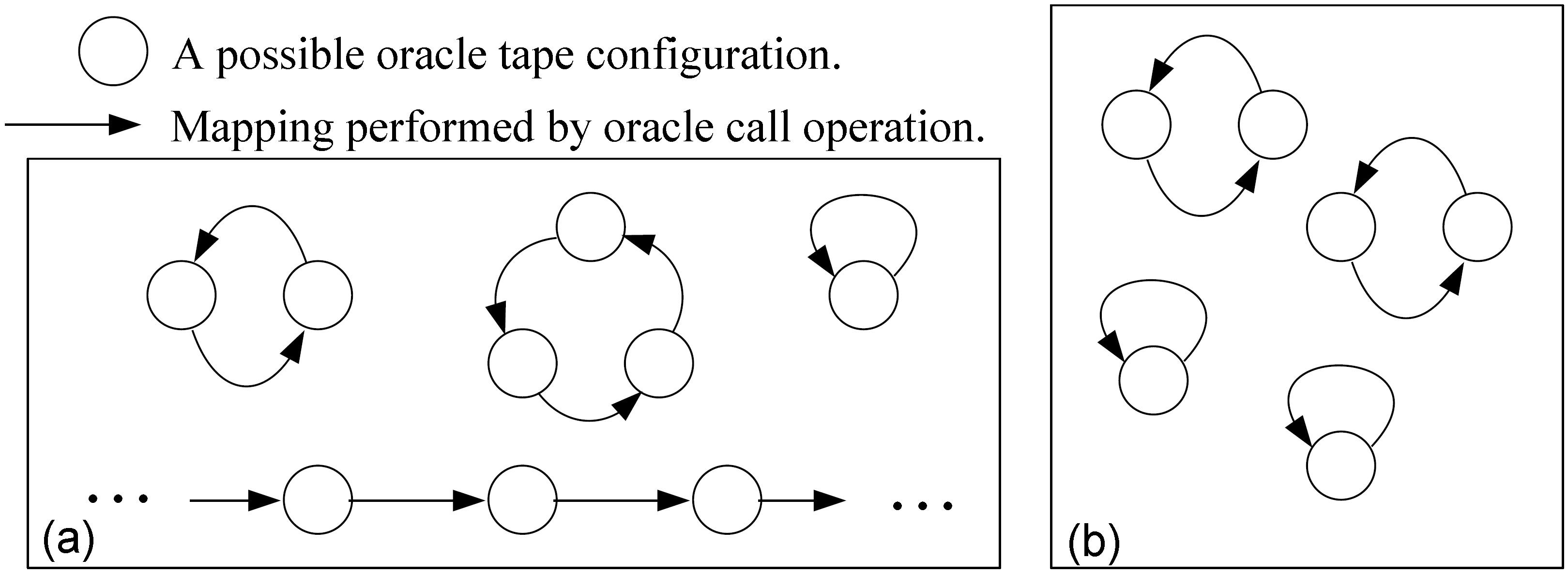}} }

\caption{\label{f:permut}Illustration of the structure of (a) a permutation oracle,
and (b) a self-reversible permutation oracle.\captionpars

In either case, the oracle call operation replaces the old contents of the oracle
tape with new contents according to a transition function \( A:\tapespace \rightarrow \tapespace  \)
that is a permutation mapping---a bijective function---over the space \( \tapespace  \)
of possible tape contents. The bijectivity of this function means that a call
to a permutation oracle is always a reversible operation. After an oracle call,
the previous oracle tape contents can be uniquely determined by applying the
inverse mapping \( A^{-1} \). In self-reversible oracles, \( A=A^{-1} \).
}

\end{figure}

\subsection{\protect\( \SP \TI \protect \)-constructibility}

In order for our proof to go through, we will need to restrict our attention
to space and time functions \( \SP (\nin ),\TI (\nin ) \) which are \defn{\(ST\)-constructible},
meaning that given any input of length \( \nin  \)\label{p:nin}, an irreversible
machine can construct binary representations of the numbers \( \SP (\nin ) \)
and \( \TI (\nin ) \) using only space \( \Atmost (\SP (\nin )) \) and time
\( \Atmost (\TI (\nin )) \). We state here without proof that many reasonable
pairs of functions are indeed \( \SP \TI  \)-constructible. For example, \( \SP =\nin ^{2} \),
\( \TI =\nin ^{3} \) can both be computed in time \( \Atmost (\log ^{2}\nin ) \)
plus \( \Atmost (\nin ) \) to count the input bits, and space \( \Atmost (\log \nin ) \)
plus \( \Atmost (\nin ) \) if we include the input.

\noheading
Next, we need some basic definitions to support the notion of incompressibility
that will be crucial to the proof of our theorem. The following definition and
lemma follow the spirit of the discussions of incompressibility in Li and Vit\'{a}nyi's
excellent book on Kolmogorov complexity \cite{LiVbook}.

\subsection{Description systems and compressibility}\label{s:descsys}

\begin{definition}
A \defn{description system} \( \descsys  \) is any function \( \funcdr {\descsys }{\bitset ^{*}}{\bitset ^{*}} \)\label{p:descsys}\label{p:setexp}
from bit-strings to bit-strings, that is, from \defn{descriptions} to the bit-strings
they describe. We say that a bit-string \( d \)\label{p:desc} \defn{describes}
bit-string \( y \)\label{p:describee} in description system \( s \) if \( s(d)=y \).
We say that a bit-string \( y \) is \defn{compressible} in description system
\( \descsys  \) if there is a shorter bit-string that describes it; \ie\ if
there exists a string \( d \) such that \( s(d)=y \) and \( |d|<|y| \), where
the notation \( |b| \)\label{p:bitstring} denotes the number of bits in bit-string
\( b \).
\end{definition}
\begin{lemma}[Existence of incompressible strings] 
For any description system \( s \),
and any string length \( \ell \in \N  \)\label{p:ell}\label{p:N}, there is
at least one bit-string \( y \) of length \( \ell  \) that is not compressible
in \( s \).
\end{lemma}
\begin{proof}[Trivial counting argument] There are \( 2^{\ell } \)
bit-strings of length \( \ell  \), but there are only \( \sum _{i=0}^{\ell -1}2^{i}=2^{\ell }-1 \)
descriptions that are shorter than \( \ell  \) bits long. Each description
\( d \) can describe at most one bit string of length \( \ell  \), namely
the string \( s(d) \) if that string's length happens to be \( \ell  \). Therefore
there must be at least one remaining bit-string \emph{y} of \emph{}length \( \ell  \)
that is not described by any shorter description.
\end{proof}

In our main proof, we will be selecting incompressible strings from a series
of computable description systems.

\subsection{Notational conventions}

In the following, we will often abbreviate the space and time function
values \( \SP (\nin ) \) and \( \TI (\nin ) \) by just \( \SP \) and
\( \TI \), respectively; likewise for other functions of \( \nin
\). For comparing orders of growth, we will use both the standard \(
\Exactly \), \( \Atmost \), \( \Atleast \), \( \Lessthan \), \(
\Morethan \) notations, and our mnemonic \( \exactly \), \( \atmost
\), \( \atleast \), \( \lessthan \), \( \morethan \) notation, defined
in table~\ref{t:oog}.
\begin{table}
{\centering \begin{tabular}{p{0.75in}cp{2.75in}}
\hline 
Cryptic &
 A more &
\\
 standard &
 mnemonic &
 Mathematical definition; \\
 notation &
 notation &
 English explanation \\
\hline 

{\raggedright\( f=\Exactly (g) \) or\\ \( f\in \Exactly (g) \)}&
 \( f\exactly g \)&
\( \exists c_{1},c_{2},n_{0}>0:\, \forall n>n_{0}:\, 0< c_{1}g(n)< f(n)< c_{2}g(n) \); \emph{f}
has the same asymptotic order of growth as \emph{g.}\\[0.1in]

{\raggedright\( f=\Atmost (g) \) or\\ \( f\in \Atmost (g) \)}&
\( f\atmost g \)&
\( \exists c,n_{0}>0:\, \forall n> n_{0}:\, 0< f(n)< cg(n) \); \emph{f}
has a lower asymptotic order of growth than \emph{g}.\\[0.1in]

{\raggedright \( f=\Atleast (g) \) or\\ \( f\in \Atleast (g) \)} &
 \( f\atleast g \)&
\( \exists c,n_{0}>0:\, \forall n> n_{0}:\, 0< cg(n)< f(n) \); \emph{f}
has a greater asymptotic order of growth than \emph{g}.\\[0.1in]

{\raggedright  \( f=\Lessthan (g) \) or\\ \( f\in \Lessthan (g) \)}&
\( f\lessthan g \)&
\( \forall c>0:\, \exists n_{0}>0:\, \forall n> n_{0}:\, 0< f(n)<cg(n) \); \emph{f}
has a strictly lower asymptotic order of growth than \emph{g}.\\[0.1in]

{\raggedright\( f=\Morethan (g) \) or\\ \( f\in \Morethan (g) \)}&
\( f\morethan g \)&
\( \forall c>0:\, \exists n_{0}>0:\, \forall n> n_{0}:\, 0< cg(n)<f(n) \); \emph{f}
has a strictly greater asymptotic order of growth than \emph{g}.\\
\hline 
\end{tabular}\par}

\caption{\label{t:oog}Asymptotic order-of-growth notation. In addition to reviewing
the standard notation, we introduce a simplified, more mnemonic notation that
will be convenient in some contexts.}

\end{table}

\section{Main theorem}\label{s:main}

\paragraph{Preliminary discussion}
In this section we prove that reversible machine models require higher asymptotic
space-time complexity on some problems than corresponding irreversible models,
if a certain new reversible black-box operation (a self-reversible oracle \emph{A})
is made available to both models. Thus, no \emph{completely} general technique
can exist for simulating irreversible machines on reversible ones with no asymptotic
overhead.

However, the new primitive operation that we defined in order to make this proof
go through is not itself physically realistic. The operation implements a computable
function, but the operation is modeled as taking constant (\( \Exactly (1) \))
time to perform independent of the size of its input, which violates physical
locality and the asymptotically very large number of steps that it would take
to compute the operation using the algorithm that corresponds directly to the
operation's definition.

Therefore, technically, even given our proof, it is still an open question whether
a perfectly efficient simulation technique might still exist that works in the
case of reversible machines simulating irreversible machines that are composed
only of primitives that are physically realistic. 

Incidentally though, if one wishes to progress to \emph{complete} physical realism,
then to be completely fair, one should take into account the physical time and
space costs associated with removing the physical entropy produced by irreversible
operations from a machine, when comparing reversible and irreversible machine
models. We do this in \cite{Frank-99} and conclude that under certain reasonable
assumptions, a variety of physically realistic reversible models are actually
asymptotically strictly \emph{more} spacetime-efficient on some problems than
are the corresponding irreversible models, although an extremely large scale
of machine may be required to realize that particular theoretical benefit.

The encroaching issue of lower limits on bit energies is more important. As
we mentioned in \S\ref{s:intro}, the exact magnitude of the purely \emph{computational}
asymptotic overheads incurred by reversible operation has an important role
to play in helping to make an accurate comparison between the potential efficiency
of reversible and irreversible machine designs in particular technologies. It
is a key element that drastically affects the shape of the overall tradeoff
function between energy costs and hardware costs in partially-adiabatic machine
design spaces.

Below, we will prove our results in both oracle-relativized and non-oracle forms
for serial (uniprocessor) machines. The oracle results cover a large family
of possible asymptotic bounds on the joint space and time requirements of computations.
For all bounding functions within this family, we show that there exist an oracle
and a language such that the language is decidable within the given bounds by
serial machines that can query the oracle only if the machines are \emph{ir}reversible.
This result is non-trivial (compared to Pin's, for example) because the individual
oracle calls are themselves reversible and easy to undo.

In section~\ref{s:nonrel}, a similar result, not involving an oracle, covers
cases where the space bound is much smaller than the length of the randomly
(and reversibly) accessible input. Corollaries to both the oracle and non-oracle
results give loose lower bounds on the amount of extra space a reversible machine
will require to decide the language within the same time bounds as the irreversible
machine, although one should keep in mind that this approach of meeting the
time bounds will not necessarily minimize the real costs corresponding to the
space-time \emph{product}.

Another contribution of our proof is to illustrate ways to use incompressibility
arguments in analyzing reversible machines. It is conceivable that similar techniques
might increase the range of reversible and irreversible space-time complexity
classes that we can separate without resorting to the oracle.

\subsection{Statement of main theorem}\label{s:st-thm}

\begin{theorem}\label{thm:revsep}
\emph{(Relative separation of reversible and irreversible space-time complexity classes.)}
Let \( \SP ,\TI \) be any two non-decreasing functions over the
non-negative integers. Then both of the following are true:
\begin{enumerate}
\item[(a) ] If \( \SP \atleast \TI  \) or \( \TI \atleast2 ^{\SP } \), then \( \RST {\SP }{\TI }^{O}=\ST {\SP }{\TI }^{O} \)
for any self-reversible oracle \( O \).
\item[(b) ] If \( \SP \lessthan \TI \lessthan2 ^{\SP } \), and if \( \SP ,\TI  \)
are \( \SP \TI  \)-con\-struct\-ible, then there exists a computable, self-reversible
oracle \( A \) such that \( \RST {\SP }{\TI }^{A}\neq \ST {\SP }{\TI }^{A} \).
\end{enumerate}
\end{theorem}
\begin{proof}
\begin{demo}{Part (a)}
(Cases \( \SP \atleast \TI  \) and \( \TI \atleast2 ^{\SP } \).) First, if
\( \SP \morethan \TI  \), then obviously we have both \( \RST {\SP }{\TI }^{O}=\RST {\TI }{\TI }^{O} \)
and \( \ST {\SP }{\TI }^{O}=\ST {\TI }{\TI }^{O} \) simply because in time
\( \TI  \) no more than \( \SP \exactly \TI  \) memory cells can be accessed
on a machine that performs \( \Exactly (1) \) operations per time step. Similarly,
if \( \TI \morethan2 ^{\SP } \), then \( \RST {\SP }{\TI }^{O}=\RST {\SP }{2^{\SP }}^{O} \)
and \( \ST {\SP }{\TI }^{O}=\ST {\SP }{2^{\SP }}^{O} \), because no computation
using only \( \SP  \) bits of memory can run for more than \( 2^{\SP } \)
steps without repeating. So part (a) reduces to proving \( \RST {\SP }{\TI }^{O}=\ST {\SP }{\TI }^{O} \)
only for the case where \( \SP \exactly \TI  \) or \( \TI \exactly 2^{\SP } \).

From here, the result follows due to the existing relativizable simulations.
When \( \SP \exactly \TI  \), Bennett's simple reversible simulation technique
\cite{Bennett-73} can be applied because it takes time \( \Atmost (\TI ) \)
and space \( \Atmost (\TI ) \). Similarly, when \( \TI \exactly 2^{\SP } \)
the simulation of Lange \etal{}\ {}\cite{Lange-etal-97} can be used because
it takes time \( \Atmost (2^{\SP }) \) and space \( \Atmost (\SP ) \). Both
techniques can be easily seen to relativize to any self-reversible oracle \( O \).
Thus, in both cases, any irreversible machine can be simulated reversibly in
\( \Atmost (\TI ) \) and space \( \Atmost (\SP ) \), and therefore \( \RST {\SP }{\TI }^{O}=\ST {\SP }{\TI }^{O} \).
\end{demo}

\begin{demo}{Part (b)}
(Case \( \SP \lessthan \TI \lessthan2 ^{\SP } \).) 
Here, we give only an outline of the full proof of part (b), which will be fleshed out in 
\S\S\ref{s:specdef}--\ref{s:lemma} below.  \textbf{Proof outline:} We will
construct \( A \) to be a permutation oracle that can be interpreted as specifying
an infinite directed graph of nodes with outdegree at most 1. We will also define
a corresponding language-recognition problem, which will be to report the contents
of a node that lies \( \TI /\SP  \) nodes down an incompressible linear chain
of nodes that have size-\( \SP  \) identifiers, starting from a node that is
determined by the input length. The oracle will be explicitly constructed via
a diagonalization, so that for each possible reversible machine, there will
be a particular input for which our oracle makes that particular reversible
machine take too much space or else get the wrong answer. In the cases where
the reversible machine takes too much space, we will prove this by equating
the machine's operation with the {}``pebble game{}'' for which Li and Vit\'{a}nyi
\cite{Li-Vitanyi-96b} have already proven lower bounds, and by showing that
if the machine does not take too much space, then we can build a shorter description
of the chain of nodes using the machine's small intermediate configurations,
thus contradicting our choice of an incompressible chain.

\end{demo}

\end{proof}

Before we can develop the proof of part (b) in full detail, we need
some more definitions specialized to our needs.

\subsection{Specialized definitions}\label{s:specdef}

\begin{definition}
A \defn{graph oracle} is a self-reversible permutation oracle with the
following property: There exists a partial function \( \funcdr
{f}{\bitset ^{*}}{\bitset ^{*}} \)\label{p:f-succfunc}, called a
\defn{successor function}, such that for any bit string (node) \( b\in
\bitset ^{*} \) for which \( f \) is defined, the oracle's permutation
function maps the tape contents \( b \) to the tape contents \( b\hash
f(b) \), and also maps \( b\hash f(b) \) back to \( b \), where
\hash{}\label{p:hash} is a special separator character in the oracle
tape alphabet. For all tape contents \( x \) not of either of these
forms, the oracle's permutation function maps them to themselves. See
fig.~\ref{f:grapho}.
\end{definition}

\begin{figure}
\centerline{\resizebox*{4in}{!}{\includegraphics{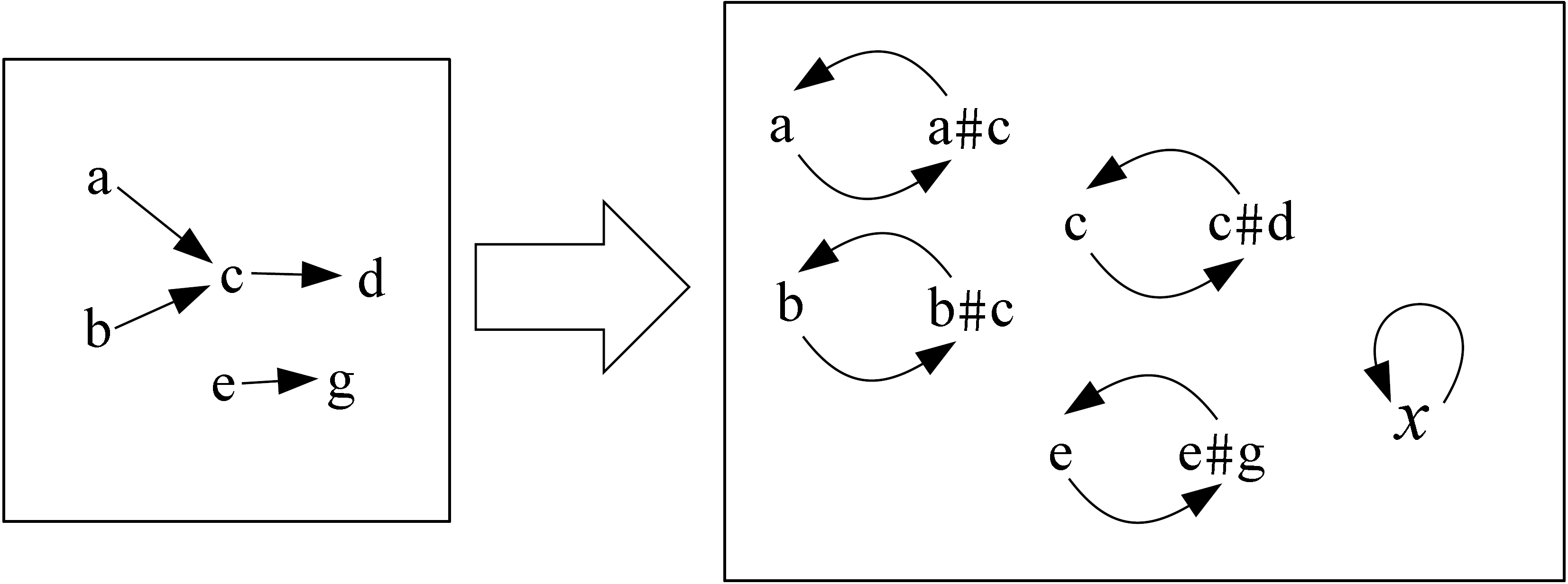}} }

\caption{\label{f:grapho}Encoding outdegree-1 directed graphs in self-reversible permutation
oracles. Letters stand for nodes represented as bit-strings, except for \protect\( x\protect \)
which represents any other bit-string not explicitly shown. The \texttt{\#}
is a special separator character.
\captionpars

On the left, we show an example of an outdegree-1 directed graph with bit-string
nodes abbreviated a,b,c,d,e,g. The graph function \( f \) gives the successor
of each node: \( f({\textrm{a}})={\textrm{c}} \), \( f({\textrm{c}})={\textrm{d}} \),
etc. This \( f \) is a partial function; \emph{e.g.}\ \( f({\textrm{d}}) \)
is undefined. For each edge in this graph, there is a corresponding pair of
strings that are mapped to each other by the self-reversible oracle. To represent
the edge \( \text {a}\rightarrow \text {c} \), for example, the permutation
oracle maps tape contents {}``a{}'' to {}``a\texttt{\#}c{}'' and maps {}``a\texttt{\#}c{}''
back to {}``a{}''. Any other string \( x \) (including those for terminal
nodes of the graph) is simply mapped to itself. In this way the permutation
oracle allows easily and reversibly looking up a node's successor, or uncomputing
a node's successor given the node and its successor. But finding a node's predecessor(s),
given just the node itself, is designed to be hard. Thus the oracle call resembles
the reversible computation of a {}``one-way{}'' invertible function that is
easy to compute, but whose inverse is difficult to compute.}

\end{figure}

\begin{remark}
The name ``graph oracle'' for this concept is really over-general; our
graph oracles are capable of embodying only graphs of a special type,
namely directed graphs in which all nodes are named by bit-strings and
have out-degree 1.  The unique node that is adjacent from node $q$ is
given by the successor function $f(q)$.
\end{remark}

Given that we will be working only with graph oracles, we can now
specify an oracle by specifying just the successor function \( f \)
that it embodies.  But before we actually construct the special oracle
\( A \) that proves theorem~\ref{thm:revsep}, let us define, relative
to \( A \), the language that we claim separates \( \RST {\SP }{\TI
}^{A} \) from \( \ST {\SP }{\TI }^{A} \).

\begin{definition}
\label{def:seplang}
Given two \( \SP \TI  \)-\discretionary{}{}{}con\-struct\-i\-ble
functions \( \SP (n) \), \( \TI (n) \), and graph oracle \( A \) with successor
function \( f \), we define the \defn{separator language} \( L(A) \)\label{p:L}
to be the language decided by the irreversible machine described by algorithm 1 in figure~\ref{fig:sepalg}.
\end{definition}

The algorithm is essentially this: Given a string of length \( n \),
construct a string of zeros of length \( \SP (n) \). Treat this string
as a node identifier, and use oracle queries to proceed down its chain
of successors for up to \( \lfloor \TI /\SP \rfloor \) nodes. Finally,
return the first bit of the final node's bit-string identifier.

\begin{figure}
\begin{algorithm}[SEPARATOR($w$)]
\label{alg:sep}
Given input string \(w\)\label{p:w},
\    Let $n$ = $|w|$; compute \(\SP=\SP(n),\TI=\TI(n)\).
\    Let bit-string \(b=\zerobit^{\SP}\).
\    Repeat the following, \(t=\lfloor \TI/\SP\rfloor\) times:
\        Write \(b\) on the oracle tape, and call the oracle $A$.
\        If r{e}sult is of the form \(b\){\texttt{\#}}\(c\), with \(c\) a bit-string,
\             assign \(b \leftarrow c\) (note that \(c=f(b)\)),
\        else, quit loop early.
\    Accept iff \(b[0]=\onebit\).
\end{algorithm}
\caption{Irreversible algorithm defining the language $L(A)$ that
separates \ST{\SP}{\TI} from \RST{\SP}{\TI}, relative to our
reversible oracle $A$.  The essence of this algorithm is simply to
interpret $A$ as a graph oracle, construct an initial node (which is
dependent on the input string), and follow the directed path leading
away from the initial node for a certain number of steps.
\label{fig:sepalg}
}
\end{figure}

We will be explicitly constructing the successor function \( f \) so
that it always returns a string of the same length as its input. Given
the corresponding oracle, algorithm 1 obviously requires only space \(
\Atmost (\SP ) \) and time \( \Atmost (\TI ) \) on on irreversible
machine in any standard serial model of computation. (Recall that \(
\SP ,\TI \) are \( \SP \TI \)-constructible.)  Therefore the language
\( L(A) \) will be in the class \( \ST {\SP }{\TI }^{A} \).

In \S\ref{s:construct}, we will show how to construct \( f \) so that
the language \( L(A) \) will not be computable by any reversible
machine that takes space \( \Atmost (\SP ) \) and time \( \Atmost (\TI
) \). The way we will do this is to make each of the node identifiers
be a different incompressible string. Intuition suggests that the only
way to decide \( L(A) \) is to actually follow the entire chain of
nodes, to see what the final one is. But having obtained a node's
successor, the reversible machine cannot easily get rid of its
incompressible records of the prior nodes. The graph oracle provides
no convenient way to compute \( f^{-1} \) and find a node's
predecessor, even if the successor function \( f \) happens to be
invertible. Thus (as we will show) the reversible machine will tend to
accumulate records of previous nodes, of size \( \SP (\nin ) \) each,
and thus, for sufficiently long enough chains, it will take more than
a constant factor times \( \SP (\nin ) \) space. The reversible
machine could conceivably find and uncompute predecessor nodes by
searching them all exhaustively, but this would take too much time.

The situation with this oracle language resembles the non-oracle problem of
iterating a one-way function, \ie{} an invertible function whose inverse much
is harder to compute than the function itself (\eg, MD5). Public-key cryptography
depends on the (unproven, but empirically reasonable) assumption that some functions
are one-way. The same assumption might allow us to show that \( \RST {\SP }{\TI }\neq \ST {\SP }{\TI } \)
without an oracle, by using a one-way function instead.

\subsection{Oracle construction}\label{s:construct}

We now construct a particular oracle \( A \) (given any appropriate
\TI,\SP) and prove that \( L(A)\notin \RST {\SP }{\TI }^{A} \).

First, fix some standard enumeration of all reversible oracle-querying
machines.  The enumeration is possible because reversible Turing
machines, for example, can be characterized by local syntactic
restrictions on their transition function, as in Lange \etal{}, so we
can enumerate all machines and pick out the reversible ones. Let \(
(\machine _{1},c_{1}),(\machine _{2},c_{2}),\ldots \)\label{p:enum} be
this enumeration dovetailed together with an enumeration of the
positive integers. If a given machine always runs in space \( \Atmost
(\SP ) \) and time \( \Atmost (\TI ) \) then it will eventually appear
in the enumeration paired with a large enough \( c_{i} \) so that the
machine \( \machine _{i} \) takes space less than \(
c_{i}+c_{i}\SP(\nin ) \) and time less than \( c_{i}+c_{i}\TI (\nin )
\) for any input length \( \nin \).

\begin{figure}
\centerline{\resizebox*{0.8\textwidth}{!}{\includegraphics{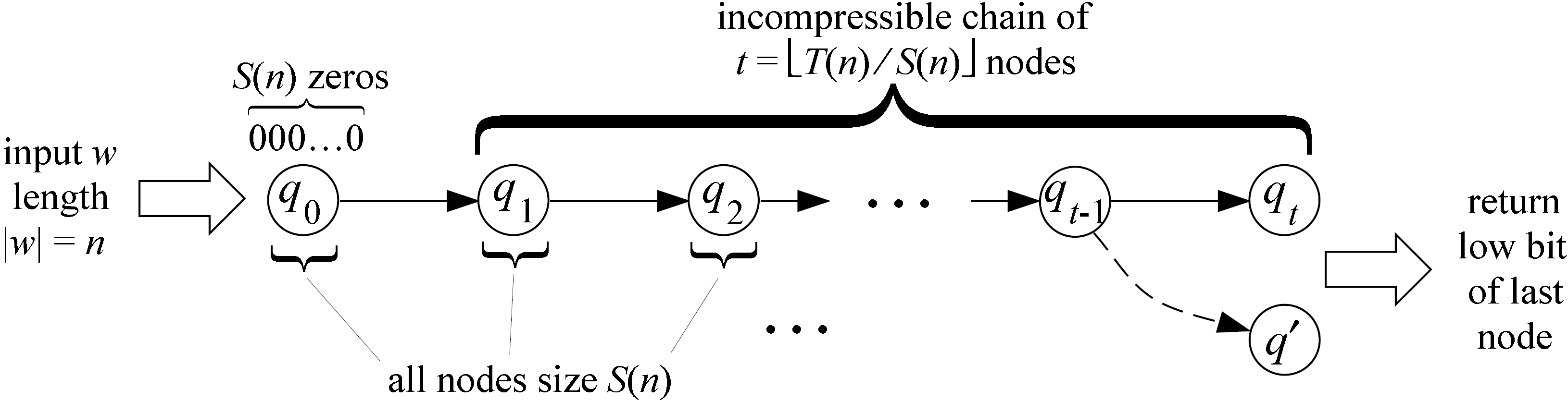}} }

\caption{\label{f:chain}The problem graph defined by our oracle for inputs of size
\protect\( n\protect \). The {}``correct answer{}'' is just the first bit
of the final node \protect\( q_{t}\protect \). If the reversible machine \protect\( M_{i}\protect \)
that we are trying to foil happens to get the right answer, but never asks for
the successor of node \protect\( q_{t-1}\protect \), we redefine \protect\( q_{t-1}\protect \)'s
successor to be a new node \protect\( q'\protect \) having a different initial
bit. }

\end{figure}

We will construct the oracle \( A \) so that each machine \( \machine _{i} \)
will fail to decide \( L(A) \) within these bounds. When considering \( \machine _{i} \),
\( f(q) \) will have already been specified for all oracle queries \( q \)\label{p:query}
asked by machines \( \machine _{1},\machine _{2}\ldots ,\machine _{i-1} \)
when given certain inputs of lengths \( n_{1},n_{2},\ldots ,n_{i-1} \)\label{p:ni},
respectively. Now, choose \( n_{i} \) (henceforth called \( n \)\label{p:n}),
the input length for which our oracle definition will foil \( \machine _{i} \),
to be such that \( \SP (n) \) is greater than the maximum length \( z \)\label{p:z}
of any of those earlier machines' oracle queries. Some other lower bounds on
the size of \( n \) will be mentioned as we go along, and are summarized in
table~\ref{t:n-constraints}.
\begin{table}
{\centering \begin{tabular}{cc}
\hline 
Constraint on \( n_{i} \)&
Introduced on\\
\hline 
\( \SP (n_{i})>z \)&
p.~\pageref{p:z}\\
\( |(j,k,x')|<|x| \) &
p.~\pageref{p:xprime}\\
\( \frac{1}{2}2^{\SP (n_{i})}>c_{i}+c_{i}\TI (n_{i}) \)&
p.~\pageref{p:altsavail}\\
\( |(j,\Delta \tau _{j},k_{j})|<\frac{1}{4}\SP (n_{i}) \)&
p.~\pageref{p:htriples}\\
\( t(n_{i})\geq 2^{4(c_{i}+1)} \)&
p.~\pageref{p:lotsapebs}\\
\( \SP (n_{i})\geq c_{i} \)&
p.~\pageref{p:spacebeatsconst}\\
\hline 
\end{tabular}\par}

\caption{\label{t:n-constraints}Constraints on the input length \nsubi\ chosen to foil
machine \protect\( M_{i}\protect \) running within bounds determined by \SP,
\TI, and \protect\( c_{i}\protect \).}

\end{table}

\ed{The proof can be simplified by requiring \(n>z\) above, rather than only
\(\SP(n)>z\), and then just letting the input \(w\) be \(q_{0}\) directly, rather
than computing \(q_{0}\) from it.  Then, \(\SP\) does not have to be
constructible or even computable (unless we want the oracle to be also)
since the first oracle call gives us \(\SP\) from \(w\).  Making this change
wouldn't be too hard.}

Later we will specify a description system \( \descsys _{i} \)\label{p:descsysi},
summarized in table~\ref{t:ds}, based on \( \machine _{i} \), \( c_{i} \),
the value of \( n \), and all the \( f(q) \) values defined so far (for bit-strings
smaller than \( \SP (n) \)). The description system will be a total computable
function, \ie, there is an algorithm that computes \( \descsys _{i}(d) \) for
any \( d \) and always halts. We will use this description system to define
\( f(q) \) for bit-strings \emph{q} of \emph{}length \( \SP (n) \), as follows:
\begin{table}
{\centering \begin{tabular}{lc}
\hline 
Description format&
Explained on\\
\hline 
\( (j,k,x') \)&
p.~\pageref{p:xprime}\\
\( (j,x') \)&
p.~\pageref{p:nozeros}\\
\( (C_{\tau },D,x', \) \emph{h} triples \( (j,\Delta \tau _{j},k_{j}) \), extra
bits)&
p.~\pageref{p:trickyone}\\
\( (j,\Delta \tau _{j},k_{j},x') \)&
p.~\pageref{p:simfromstart}\\
\hline 
\end{tabular}\par}

\caption{\label{t:ds}Description formats needed in description system \si.}

\end{table}

Let \( x \)\label{p:x} be a bit-string of length \( \TI (n) \) that is incompressible
in description system \( s_{i} \) (to be defined as we go along). This \( x \)
will be used as the sequence of size-\( \SP (n) \) node identifiers that will
define our graph for inputs of size \( n \).

Break \( x \) up into a sequence of \( t(n)\equiv \lfloor \TI /\SP
\rfloor \)\label{p:t} bit-strings of length \( \SP (n) \) each; call
these our graph nodes or \emph{query strings\/{}} \( q_{1},\ldots
,q_{t} \)\label{p:qj}. (Due to the floor operation, up to \( \SP -1 \)
bits may be left over; these aren't used in any query strings.)  We
will design our description system \( s_{i} \) so that all the \(
q_{j} \)'s\label{p:j} must be different.  We accomplish this by
allowing descriptions of the form \( (j,k,x')
\),\label{p:k-qindex}\label{p:xprime} where \( j \) and \( k \) are
the indices of two equal nodes \( q_{j}=q_{k} \), \( j<k \), and \( x'
\) is \( x \) with the \( q_{k} \) substring spliced out. The
description system would be defined to generate \( x \) from such a
description by simply looking up the string \( q_{j} \) in \( x' \)
and inserting a copy of it in the \( k \)th position. The indices \( j
\) and \( k \) would take \( \Atmost (\log (\TI /\SP )) \) space,
which is \( \Atmost (\log \TI ) \) space, which is \( o(\SP ) \)
space, whereas we are saving \( \SP (n) \) space by not explicitly
including the repetition of \( q_{j} \). Therefore as long as \( n \)
is sufficiently large, the total length of this description of \( x \)
would be less than \( |x| \). With \( x \) being incompressible in a
description system that permits such descriptions, we know that \(
q_{1},\ldots ,q_{t} \) includes no repetitions.

Now we can specify exactly how the oracle defines our problem graph
for inputs of size \( n \), as follows. Define query string \(
q_{0}={\texttt {0}}^{\SP } \)\label{p:q0} (a string of \( \SP \)
\zerobit{}-bits). Provisionally, set \( f(q_{j-1})=q_{j} \) for all \(
1\leq j\leq t \). These assignments are possible since all the \(
q_{j} \)'s are different, as we just proved. (They also must be
different from \( q_{0} \), but this is easy to ensure as well, using
descriptions of the form \( (j,x') \)\label{p:nozeros}.)  Given these
assignments, all strings of length \( n \) are in the language \( L(A)
\) if and only if \( q_{t}[0]={\texttt {1}} \) (where \( q_{t}[0] \)
means the first bit of \( q_{t} \)), due to the earlier definition of
\( L(A) \).  (Definition~\ref{def:seplang}.)

Suppose temporarily that our oracle definition were completed by
letting \( f \) remain undefined over all strings \( w \) for which we
have not yet specified \( f(w) \). (\Ie, let \( A(w)=w \) for these
strings.) Under that assumption, simulate \( M_{i} \)'s behavior on
the input \( {\texttt {0}}^{n} \). If \( M_{i} \) runs for more than
\( c_{i}+c_{i}T \) steps, then it takes too much time, and we are
through addressing it. Otherwise, \( M_{i} \) either accepts (\(
{\texttt {1}} \)) or rejects (\( {\texttt {0}} \)). If this answer is
different from \( q_{t}[0] \), then \( M_{i} \) already fails to
accept the language \( L(A) \), and we are through with it.

Alternatively, suppose \( M_{i} \)'s answer is correct with the given \( q_{j} \)'s
and it halts within \( c_{i}+c_{i}T \) steps. But now, suppose that \( M_{i} \)
never asked any query that was dependent on our choice of \( f(q_{t-1}) \)
during its run on input \( {\texttt {0}}^{n} \). That is, suppose \( M_{i} \)
never asked either query \( q_{t-1} \) or query \( q_{t-1} \)\texttt{\#}\( q_{t} \).
In that case, let us change our definition of \( f(q_{t-1}) \) as follows,
to change the correct answer to be the opposite of what \( M_{i} \) gave. Let
\( q' \)\label{p:qprime} be a bit-string that is independent of all queries
made by \( M_{i} \) in that simulation, and whose first bit is the opposite
of \( M_{i} \)'s answer. To ensure such strings exist, note there are \( \frac{1}{2}2^{\SP } \)
bit-strings of length \( \SP  \) having the desired initial bit, but \( M_{i} \)
can make at most \( c_{i}+c_{i}T \) queries since that is its running time.
We know \( \TI \lessthan 2^{\SP } \), so with sufficiently large \( n \),
\( \frac{1}{2}2^{\SP }>c_{i}+c_{i}\TI  \)\label{p:altsavail}, and we can find
our node \( q' \). Now, given \( q' \), we change \( f(q_{t-1}) \) to be
\( q' \). This cannot possibly affect the behavior of \( M_{i} \) since it
never asked about \( f(q_{t-1}) \). But the correct answer is changed to the
first bit of \( q' \), the new node number \( t \) in the chain. Thus with
this new partial specification of \( f \), \( M_{i} \) fails to correctly
decide \( L(A) \), and we can go on to foil other machines.

Finally, suppose \( M_{i} \) does ask query \( q_{t-1} \). We now show
how to complete the definition of our description system \( s_{i} \),
source of our incompressible \( x \), so that if \( M_{i} \) does ask
query \( q_{t-1} \), then it must at some point take more than \(
c_{i}+c_{i}\SP \) space.

To do this, we show that \( M_{i} \) can always be interpreted as following
the rules of Bennett's reversible {}``pebble game,{}'' introduced in \cite{Bennett-89}
and analyzed by Li and Vit\'{a}nyi in \cite{Li-Vitanyi-96b}. 

\paragraph{Pebble game rules}\label{p:pebgamrul}
The game is played on a linear chain of nodes, which we will identify with our
query strings \( q_{1},\ldots ,q_{t} \). At any time during the game some set
of nodes is \emph{pebbled.} Initially, no nodes are pebbled. At any time, the
\emph{player\/{}} (in our case, \( M_{i} \)) may, as a move in the game, change
the pebbled vs.\ unpebbled status of node \( q_{1} \) or any node \( q_{j} \)
for which the previous node \( q_{j-1} \) is pebbled. Only one such move may
be made at a time.

The idea of the pebbled set is that it corresponds to the set of nodes that
is currently {}``stored in memory{}'' by \( M_{i} \). (We will show how to
make this correspondence explicit.) We will show that pebbling or unpebbling
node \( q_{j} \) will require querying the oracle with query string \( q_{j-1} \)
or \( q_{j-1} \)\texttt{\#}\( q_{j} \), respectively. The goal of the pebble
game is to eventually place a pebble on the final node \( q_{t} \). This corresponds
to the fact (already established) that \( M_{i} \) must at some point ask query
\( q_{t-1} \), or the oracle we are constructing will foil it trivially.

\begin{figure}
\centerline{\resizebox*{0.8\textwidth}{!}{\includegraphics{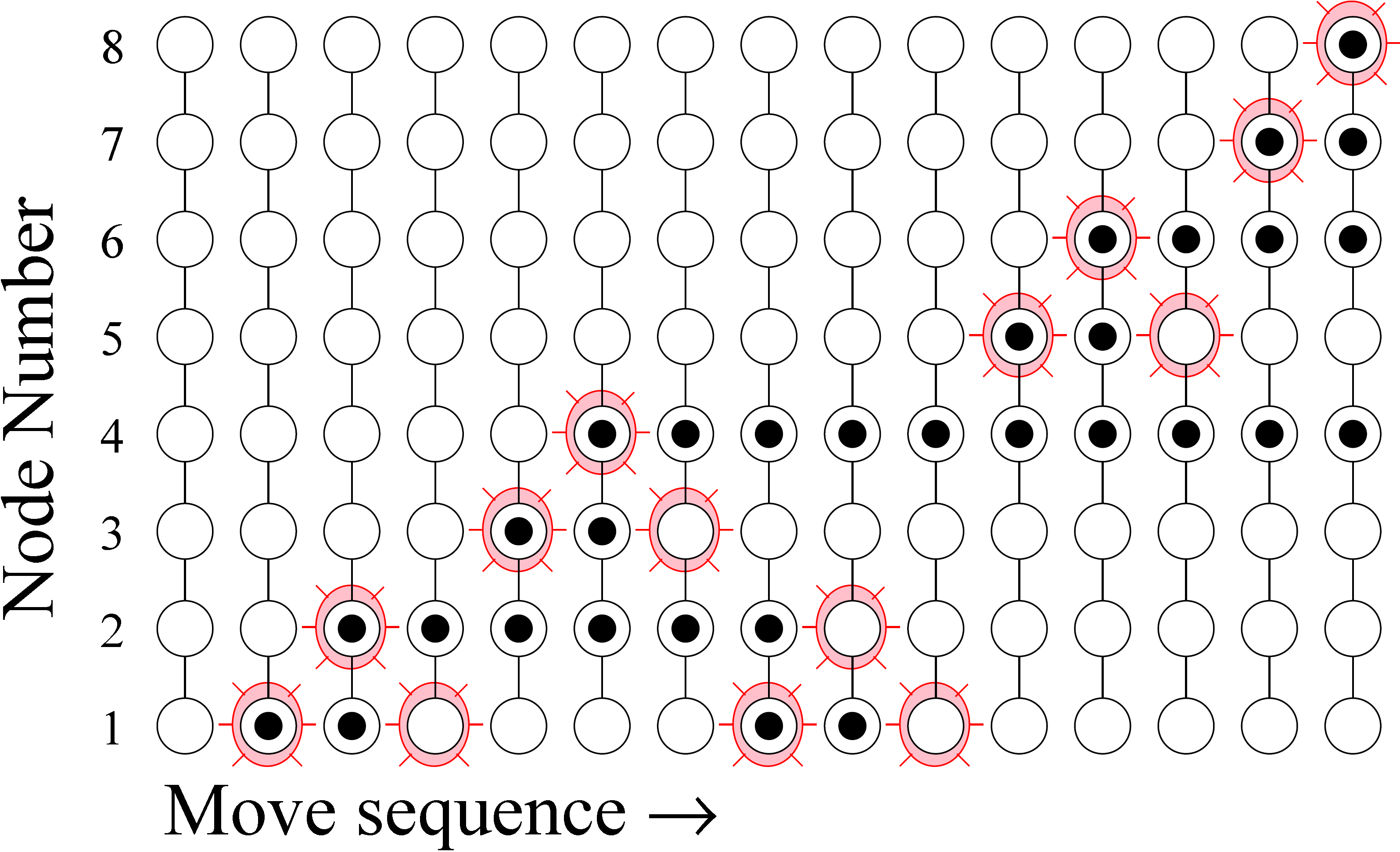}} }

\caption{\label{f:game}Bennett's reversible pebble game strategy. Highlights point
out the move made at each step. (Compare with fig.~\ref{f:bennett}(a), page~\pageref{f:bennett},
rotated \protect\( 90^{\circ }\protect \).)
\captionpars

A node \( q_{j} \) can be pebbled or unpebbled only if it is node \( q_{1} \)
or if the previous node \( q_{j-1} \) is pebbled. The strategy invented by
Bennett \cite{Bennett-89}, illustrated here, was shown by Li and Vit\'{a}nyi
to be optimal \cite{Li-Vitanyi-96a} in terms of the number of pebbles required.
But even with this optimal strategy, to pebble node \( 2^{k} \) we must at
some time have more than \( k \) nodes pebbled. In this example, we reach node
\( 2^{3}=8 \) but must use 4 pebbles to do so. (After pebbling node 8, we can
remove all pebbles by undoing the sequence of moves.) The fact that a constant-size
supply of pebbles can only reach upwards along the chain a constant distance
is crucial to our proof. }
\end{figure}

Li and Vit\'{a}nyi's analysis of the pebble game \cite{Li-Vitanyi-96b} showed
that no strategy can win the game for \( 2^{k} \)\label{p:kpebbles} nodes
or more without at some time having more than \( k \) nodes pebbled at once.
We will show that our machine \( M_{i} \) and its space usage can be modeled
using the pebble game, so that for some sufficiently large \( n \), the space
required to store the necessary number of pebbled nodes will exceed \( M_{i} \)'s
allowable storage capacity \( c_{i}+c_{i}\SP  \).

For the oracle \( A \) as defined so far, consider the complete sequence of
configurations of \( M_{i} \) given input \( {\texttt {0}}^{n} \), notated
\( C_{0},C_{2},\ldots ,C_{\TI' } \)\label{p:Ctau}\label{p:TIprime}, where
\( \TI' \leq c_{i}+c_{i}\TI  \) is \( M_{i} \)'s total running time, in terms
of the number of primitive operations (including oracle calls) performed.

Now, we need a couple of slightly more complex definitions.

\begin{definition}
\emph{(Previous and next queries involving a node.)}
For any time point \( \tau  \)\label{p:tau}, where \( 0\leq \tau \leq \TI'  \),
and for any node \( q_{j} \) in the chain of nodes \( q_{1},\ldots ,q_{t} \),
define \emph{the previous query involving\/ \( q_{j} \)} (written \( \textsc{prev}(q_{j}) \))\label{p:prev}
to mean the most recent oracle query in \( M_{i} \)'s history before time \( \tau  \)
in which the query string (the one that is present on the oracle tape at the
start of the query) is either $q_{j-1}$, $q_{j-1}{\texttt {\#}}q_{j}$, $q_{j}$, or $q_{j}{\texttt {\#}}q_{j+1}$.
There may of course be no such query, in which case \( \textsc{prev}(q_{j}) \)
does not exist. Similarly, define \emph{the next query involving\/ \( q_{j} \)}
(written \( \textsc{next}(q_{j}) \))\label{p:next} to mean the most imminent
such query in \( M_{i} \)'s future after time \( \tau  \).
\end{definition}

\begin{definition}
\label{def:pebtimes}
\emph{(A node being pebbled at a point in time.)}
Node \emph{\( q_{j} \) is pebbled at time\/ \( \tau  \)} iff at time \( \tau  \)
either:
\begin{enumerate}
\item[(a)]
\( \textsc{prev}(q_{j}) \) exists and is either
\begin{enumerate}
\item[(a.1)] \( q_{j-1} \),
\item[(a.2)] \( q_{j} \), or 
\item[(a.3)] \( q_{j} \)\texttt{\#}\( q_{j+1} \), or
\end{enumerate}
\item[(b)]
\( \textsc{next}(q_{j}) \) exists and is
\begin{enumerate}
\item[(b.1)] \( q_{j} \), 
\item[(b.2)] \( q_{j} \)\texttt{\#}\( q_{j+1} \), or
\item[(b.3)] \( q_{j-1} \)\texttt{\#}\( q_{j} \).
\end{enumerate}
\end{enumerate}
(With the exception that the
final node \( q_{t} \) is only considered pebbled in cases (a.1) and (b.3).)
\end{definition}

Note that this definition implies that \( q_{j} \) is \emph{not}
pebbled iff both \( \textsc{prev}(q_{j})=q_{j-1}{\texttt {\#}}q_{j} \)
(or nonexistent) and \( \textsc{next}(q_{j})=q_{j-1} \) (or
nonexistent).

\begin{figure}
\centerline{\resizebox*{0.75\textwidth}{!}{\includegraphics{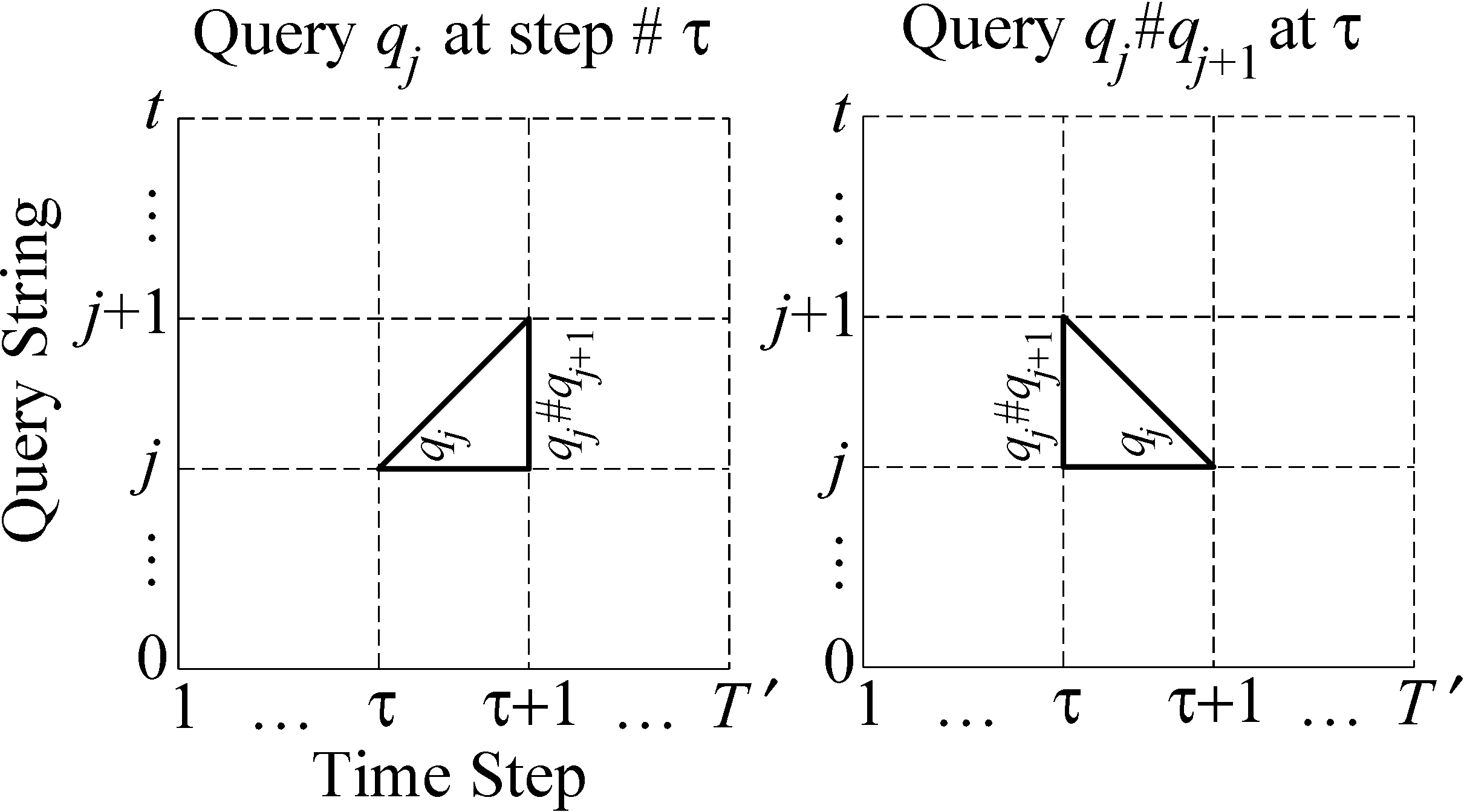}} }

\caption{\label{f:triang}Triangle representation of oracle queries.
\captionpars

The shape and direction of the triangle is meant to evoke the fact that at the
times just before and after an oracle query, the oracle tape contains the shorter
string \( q_{j} \) at one of the times, and the longer string \( q_{j} \)\texttt{\#}\( q_{j+1} \)
at the other time. The set of triangles defines the set of pebbled nodes at
any time, as illustrated in figure \ref{f:pebble}.}

\end{figure}

\begin{figure}
\centerline{\resizebox*{0.8\textwidth}{!}{\includegraphics{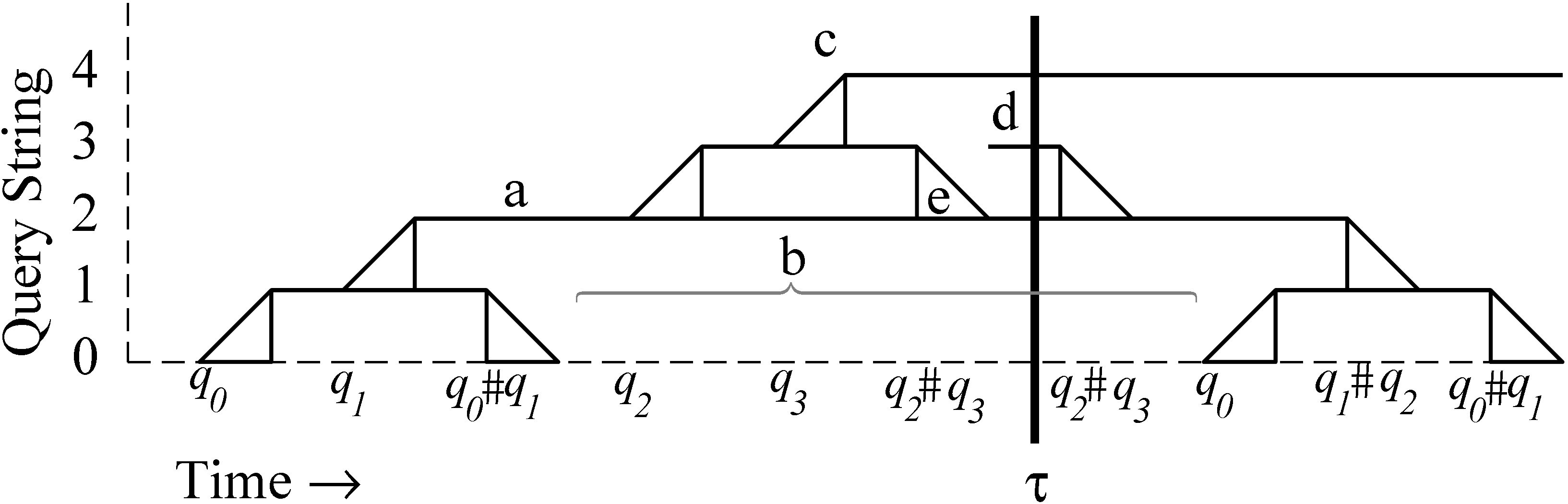}} }

\caption{\label{f:pebble}Visualizing the definition of the set of pebbled nodes.
\captionpars

The times at which a node is pebbled (indicated by solid horizontal lines on
the chart) are determined, by definition, solely by the identities and timing
of oracle queries and the corresponding arrangement of triangles (see fig.~\ref{f:triang})
on the chart. Each vertex of a triangle generates a line of pebbled times for
the corresponding node, extending horizontally away from the triangle until
it encounters another triangle. (Except query string 0 is never pebbled, because
it is not considered to be a node.)

The above example shows a pattern of queries similar to the one that would occur
if one tried to apply Bennett's \cite{Bennett-89} optimal pebble game strategy.
(Compare with figs.~\ref{f:game} and \ref{f:bennett}.)

Node 2 is considered pebbled at time (a) both because of the previous and next
queries (triangles) involving node 2. Node 1 is not pebbled at times (b) because
the previous and next queries are \( q_{0}{\texttt {\#}}q_{1} \) and \( q_{0} \)
respectively. Node 4 is pebbled at all times after (c) because even though there
is no next query involving node 4, the previous query involving node 4 exists
and is of the right form (\( q_{3} \)). Node 3 is pebbled at time (d) because
although the previous query (e) is of the wrong form (\( q_{2}{\texttt {\#}}q_{3} \)),
the next query is okay.

Query (e) does not change the set of pebbled nodes and so is not considered
to be a move in the pebble game. All the other queries are considered to be
pebbling or unpebbling moves in the pebble game, depending on the direction
of the corresponding triangle.

In the machine configuration \( C_{\tau } \) at time \( \tau  \), nodes 2,
3, and 4 are pebbled. But note that the query string for node 2 can be found
by simulating the machine backwards from time \( \tau  \) until query (e),
and reading \( q_{2} \) off of the oracle tape. And if \( q_{3} \) is given,
we can continue simulating backwards until we get to time (c), and read \( q_{4} \)
off the oracle tape as well. The ability to perform this sort of simulation,
for any arrangement of triangles, either forwards or backwards in time as needed
to find out more than a constant number of the pebbled nodes is what makes our
incompressibility argument work.}

\end{figure}

Figure \ref{f:pebble} illustrates the intuition behind this definition using
the graphical notation introduced in fig.~\ref{f:triang}. This graphical notation
is especially nice because it evokes the image of playing the pebble game or
running Bennett's algorithm (compare fig.~\ref{f:pebble} with figs.~\ref{f:game}
and \ref{f:bennett}).

The times at which a node is to be considered {}``pebbled{}'' during a
machine's execution are indicated by the solid horizontal lines on
\ref{f:pebble}. These times are determined, according to
definition~\ref{def:pebtimes} above, solely by the arrangement of
triangles (representing oracle queries, see fig.~\ref{f:triang}) on
the chart. Each vertex of a triangle generates a line of pebbled times
for the corresponding node, extending horizontally away from the
triangle until it hits another triangle.  Query string 0 is never
considered pebbled because it is not considered to be a node.

\subsection{Main Space-Bounding Lemma}\label{s:lemma}

Let \( p \)\label{p:p} denote the number of distinct nodes out of \( q_{1},\ldots ,q_{t} \)
that are pebbled at time \( \tau  \). We now lower bound the size of \( C_{\tau } \),
\emph{i.e.}\ \( M_{i} \)'s space usage at time \( \tau  \).

\begin{lemma}[Minimum space required to pebble \protect\( p\protect \) nodes]
\label{lemma:pebspace}
Given the preceding definitions, \( |C_{\tau }|>\frac{1}{4}p\SP . \) 
\end{lemma}

\begin{remark}
The constant \( \frac{1}{4} \) here is somewhat arbitrary, and with
straightforward generalization of the below proof this constant could
instead be replaced by \( \frac{1}{2}-\epsilon \) for any constant \(
\epsilon >0 \). We conjecture, but have not proven, that it could also
be replaced by any constant \( 1-\epsilon \).)
\end{remark}

\begin{proof}
Suppose \( C_{\tau } \) were no larger than \( \frac{1}{4}p\SP  \) bits. Then
we can show that \( x \) (the sequence of all \( q_{j} \)'s) is compressible
to a shorter description \( d \), which we will now specify. Our description
system \( s_{i} \) will be defined to process descriptions of the required
form.

First, note that for each node \( q_{j} \) that is pebbled at time \( \tau  \),
that node is pebbled either because of the previous query involving \( q_{j} \),
because of the next query involving \( q_{j} \), or both. Therefore, either
at least \( \frac{1}{2}p \) nodes are pebbled because of their previous query,
or at least \( \frac{1}{2}p \) nodes are pebbled because of their next query.
Let \( D \)\label{p:D} be a direction (forwards or backwards) from time \( \tau  \)
in which one can find queries causing \( h\geq \frac{1}{2}p \)\label{p:h}
nodes to be pebbled.

We now specify the shorter description \( d \) that describes \( x \). It
will contain an explicit description of \( C_{\tau } \), which by our assumption
is no longer than \( \frac{1}{4}p\SP  \). It will also specify the direction
\( D \) and contain a concatenation \( x' \) of all \( t-h \) of the nodes
\( q_{j} \) (for \( 1\leq j\leq t \)) that are \emph{not\/{}} pebbled because
of queries in direction \( D \). The size of \( x' \) will be \( (t-h)\SP  \).
For each of the \( h \) nodes \( q_{j} \) that \emph{are\/{}} pebbled because
of a query in direction \( D \), the description \( d \) will contain the
node index \( j \) and an integer \( \Delta \tau _{j} \)\label{p:Deltatauj}
giving the number of steps from step \( \tau  \) to the time of the query.\label{p:trickyone}
Also we include a short tag \( k_{j} \)\label{p:kj} indicating which of the
3 possible cases of queries causes the node to be pebbled. Each of the indices
\( j \) takes space \( \Atmost (\log t)\lessthan \log \TI \lessthan \SP  \),
and similarly each \( \Delta \tau _{j} \) takes space \( \Atmost (\log \TI )\lessthan \SP  \).
The tag is constant size. Thus for sufficiently large \( n \), all \( h\leq p \)
of the \( (j,\Delta \tau _{j},k_{j}) \) tuples together take less than \( \frac{1}{4}p\SP  \)
space\label{p:htriples}. The total space so far is less than \( t\SP  \).
If \( t\SP <\TI  \), then \( x \) will contain some additional bits beyond
the concatenation of \( q_{1}q_{2}\ldots q_{t} \), in which case \( d \) includes
those extra bits as well. The total length of \( d \) will still be less than
\( \TI =|x| \), as demonstrated in table~\ref{t:fancydesc}.
\begin{table}
{\centering \begin{tabular}{cc}
\hline 
Component of description \emph{d}&
Length\\
\hline 
\( C_{\tau } \)&
\( \leq \frac{1}{4}p\SP  \)\\
\emph{D}&
1 bit\\
\( x' \)&
\( (t-h)\SP \leq (t-\frac{1}{2}p)\SP  \)\\
\emph{h} triples \( (j,\Delta \tau _{j},k_{j}) \)&
\( h\cdot \Lessthan (\SP )=\Lessthan (\frac{1}{4}p\SP -1) \)\\
extra bits&
\( \TI -t\SP  \)\\
\hline 
TOTAL&
\( <\TI =|x| \) (for sufficiently large \emph{n})\\
\hline 
\end{tabular}\par}

\caption{\label{t:fancydesc}Size accounting for the description format \emph{d} used
to prove lemma~\ref{lemma:pebspace}.}
\end{table}

We now demonstrate that the description \( d \) is sufficient to reconstruct
\( x \), and give an algorithm for doing so. The function computed by this
algorithm tells how our description system \( s \) will handle descriptions
of the form outlined above.

The algorithm will work by simulating \( M_{i} \)'s operation in direction
\( D \) starting from configuration \( C_{\tau } \), and reading the identifiers
of pebbled nodes from \( M_{i} \)'s simulated oracle tape as it proceeds. We
can figure out which oracle queries correspond to which nodes by referring to
the stored times \( \Delta \tau _{j} \) and tags \( k_{j} \). Once we have
extracted the identifiers of all nodes pebbled in direction \( D \), we print
all the nodes out in the proper order.

As an example, refer again to fig.~\ref{f:pebble}. In the machine configuration
marked at time \( \tau  \), nodes 2, 3, and 4 are pebbled. But note that the
query string for node 2 can be found by simulating the machine backwards from
time \( \tau  \) until query (e), and reading \( q_{2} \) off of the oracle
tape. And if \( q_{3} \) is known, we can continue simulating backwards until
we get to time (c), and read \( q_{4} \) off the oracle tape as well. The ability
to perform this sort of simulation, for any arrangement of triangles, either
forwards or backwards in time as needed to find out at least half of the pebbled
nodes is what makes our incompressibility argument work. The algorithm is described
and verified in more detail in \S\ref{s:alg}.

Given \( d \), the algorithm produces \( x \), and with \( n \) chosen
large enough, the length of the description will be smaller than \( x
\) itself, contradicting the assumption of \( x \)'s incompressibility
relative to \( s \).  Therefore for these sufficiently large \( n \),
all configurations in which \( p \) nodes are pebbled must actually be
larger than \( \frac{1}{4}p\SP \).  This completes the proof of
lemma~\ref{lemma:pebspace}.
\end{proof}

\paragraph{Interpreting any \protect\( M_{i}\protect \) as playing the pebble game}
Now, given the definition of the set of pebbled nodes from earlier
(defn.~\ref{def:pebtimes}), it is easy to see how \( M_{i} \)'s
execution history can be interpreted as the playing of a pebble
game. Whenever \( M_{i} \) performs a query \( q_{j} \) and node \(
q_{j+1} \) was not already pebbled immediately prior to this query, we
say that \( M_{i} \) \emph{is pebbling node\/ \( q_{j+1} \)} as a move
in the pebble game. Similarly, whenever \( M_{i} \) performs a query
\( q_{j} \)\texttt{\#}\( q_{j+1} \) and node \( q_{j+1} \) is not
pebbled immediately after this query, we say that \( M_{i} \) \emph{is
unpebbling node\/ \( q_{j+1} \).} All other oracle queries and
computations by \( M_{i} \) are considered as pauses between pebble
game moves of these two forms. For example, in fig.~\ref{f:pebble},
query (e) (the first occurrence of \( q_{2}{\texttt {\#}}q_{3} \)) is
not considered a move in the pebble game, since it doesn't change the
set of pebbled nodes as defined by definition~\ref{def:pebtimes}.

It is obvious that under the above interpretation, all moves must obey the main
pebble game rule, \ie{}\ that the pebbled status of node \( q_{j} \) can only
change if \( j=1 \) or if node \( q_{j-1} \) is pebbled during the change.
The move is a query, and the presence of the query means the node \( q_{j-1} \)
is pebbled both before and after the query, by definition~\ref{def:pebtimes},
unless \( j=1 \) (we consider \( q_{0} \) not to be a node).

To show that no nodes are \emph{initially} pebbled (another pebble game rule)
takes only a little more work. Suppose that some nodes were pebbled in \( M_{i} \)'s
initial configuration, and consider a node \( q_{j} \) out of these that is
pebbled due to the \emph{earliest} query involving any of the initially-pebbled
nodes. Then a shorter description of \( x \) (for sufficiently large \( n \))
can be given as \( (j,\Delta \tau _{j},k_{j},x') \)\label{p:simfromstart},
where \( x' \) is \( x \) with \( q_{j} \) spliced out. This description
could be processed via simulation of \( M_{i} \) to produce \( x \) in much
the same way as in lemma~\ref{lemma:pebspace}, except that this time, the starting
configuration \( C_{1} \) can be produced directly from the known values of
\( M_{i} \) and \( n \), and need not be explicitly included in the description.
Of course the description system \( s \) needs to be able to process descriptions
of this form. Then the incompressibility of \( x \) in \( s \) shows that
the assumption that \( q_{j} \) is initially pebbled is inconsistent.

Thus, \( M_{i} \) can be seen as exactly obeying all the rules of the Bennett
pebble game. Now, Li and Vit\'{a}nyi have shown \cite{Li-Vitanyi-96b} that
any strategy for the pebble game that eventually pebbles a node at or beyond
node \( 2^{k} \) must at some time have at least \( k+1 \) nodes pebbled at
once. So let us simply choose \( n \) large enough so that \( t(n)\geq 2^{k} \)
for some \( k\geq 4(c_{i}+1) \)\label{p:lotsapebs}, and also so that \( \SP \geq c_{i} \)\label{p:spacebeatsconst}.
Then at times \( \tau  \) when \( p \) is maximum, \( M_{i} \)'s space usage
is (using lemma~\ref{lemma:pebspace}) \( |C_{\tau }|>\frac{1}{4}p\SP >\frac{1}{4}k\SP \geq (c_{i}+1)\SP \geq c_{i}+c_{i}\SP . \)

The above discussion establishes that machine \( M_{i} \) takes more than space
\( c_{i}+c_{i}\SP  \) if it correctly decides membership in \( L(A) \) for
inputs of length \( n_{i}=n \) and takes only time \( c_{i}+c_{i}\TI  \),
so long as the oracle \( A \) is consistent with the definition above. Since
machine \( M_{i} \)'s behavior on the input \( {\texttt {0}}^{n} \) only depends
on the values of the successor function \( f(b) \) for bit-strings \( b \)
up to a certain size (call it \( z \)), we are free to extend the oracle definition
to similarly foil machine \( M_{i+1} \) by picking \( n_{i+1} \) so that \( \SP (n_{i+1})>z \).
If one continues the oracle definition process in this fashion for further \( M_{i} \)'s
\emph{ad infinitum}, then for the resulting oracle, it will be the case that
for any \( M_{i} \) and constant \( c_{i} \) in the entire infinite enumeration,
the machine will either get the wrong answer or take more than time \( c_{i}+c_{i}\TI  \)
or space \( c_{i}+c_{i}\SP  \) on input \( {\texttt {0}}^{n_{i}} \). Thus,
no reversible machine can actually decide \( L(A) \) in time \( \Atmost (\TI ) \)
and space \( \Atmost (\SP ) \), and so \( L(A)\notin \RST {\SP }{\TI }^{A} \).

Note that this entire oracle construction, as described, is computable. If we
are given procedures for computing \( \SP (n) \) and \( \TI (n) \), we can
write an effective procedure that, given any finite oracle query, returns \( A \)'s
response to the query. The details of the oracle construction algorithm follow
directly from the above definition of \( A \), but would be too tedious to
present here. This concludes our proof of theorem~\ref{thm:revsep}.

Note that in the above proof, we used the fact that the number of pebbles required
to get to the final node grows larger than any constant as \( n \) increases.
But the actual rate of growth can be used as well, to give us an interesting
lower bound.

\subsection{Lower Bound Corollary}\label{s:corr1}

\begin{corollary}
\label{cor:lowerbound}
(Lower bound on space for linear-time relativizable reversible simulation
of irreversible machines.) For all \( \SP \TI  \)-constructible \( \SP ,\TI  \)
and computable \( \SP'  \)\label{p:SPprime} such that \( \SP \lessthan \TI \lessthan 2^{\SP } \)
and \( \SP' \lessthan \SP \log (\TI /\SP ) \), there exists a computable, self-reversible
oracle \( A \) such that \( \RST {\SP' }{\TI }^{A}\nsupseteq \ST {\SP }{\TI }^{A} \).
\end{corollary}

\begin{proof}
(Sketch.) Essentially the same as theorem~\ref{thm:revsep} part (b),
but with \( \SP' \) in place of \( \SP \) in appropriate places. In
the last part of the proof, \( M_{i} \) is shown to take more than \(
c_{i}+c_{i}\SP' \) space by using lemma~\ref{lemma:pebspace}, together
with the fact that \( p>\lfloor \lg \lfloor \TI /\SP \rfloor \rfloor
\) pebbles are required to reach the final node.\end{proof}

This result implies that
any general linear-time simulation of irreversible machines by reversible ones
that is relativizable with respect to all self-reversible oracles must take
space \( \Atleast (\SP \log (\TI /\SP )) \).

The most space-efficient linear-time reversible simulation technique that is
currently known was provided by Bennett (\cite{Bennett-89}, p.~770), and analyzed
by Levine and Sherman \cite{LevSher90} to take space \( \Atmost (\SP (\TI /\SP )^{1/(0.58\lg (\TI /\SP ))}) \).
Bennett's simulation can be easily seen to work with all self-reversible oracles,
so it gives a relativizable upper bound on space. There is a gap between it
and our lower bound, due to the fact that the space-optimal pebble-game strategy
referred to in our proof takes \emph{more\/{}} than linear time in the number
of nodes. A lower bound on the number of pebbles used by \emph{linear\/{}} time
pebble game strategies would allow us to expand our lower bound on space, hopefully
to converge with the existing upper bound.

\section{Non-relativized separation}
\label{s:nonrel}

We now explain how the same type of proof can be applied to show a non-relativized
separation of \( \RST {\SP }{\TI } \) and \( \ST {\SP }{\TI } \) for a certain
slowly-growing space bound \SP, when inputs are accessed in a specialized way
that is similar to an oracle query, and the input size is not included in the
space usage.

\paragraph{Input framework}
Machine inputs will be provided in the form of a random-access read-only memory
\( I \)\label{p:IROM}, which may consist of \( 2^{b} \)\label{p:b-wordlen}
\( b \)-bit words for any integer \( b\geq 0 \). The length of this input
may be considered to be \( n(b)=b2^{b} \) bits; let \( b(n) \) be the inverse
of this function. The machine will have a special \emph{input access tape} which
is unbounded in one direction, initially empty, and is used for reversibly accessing
the input ROM via the following special operations.

\emph{Get input size.} If the input access tape is empty before this operation,
after the operation it will contain \( b \) written as a binary string. If
the tape contains \( b \) before the operation, afterwards it will be empty.
In all other circumstances, the query is a no-op.

\emph{Access input word.} If the input access tape contains a binary string
\( a \)\label{p:address} of length \( b \) before the operation, afterwards
it will contain the pair \( (a,I[a]) \) where \( I[a] \) is a length-\( b \)
binary string giving the contents of the input word located at address \( a \).
If the tape contains this pair before the operation, afterwards it will contain
just \( a \). Otherwise, nothing happens.

\begin{theorem}
\label{thm:nonrelsep}
\emph{(Non-relativized separation of reversible and irreversible spacetime.)}
For models using the above input framework, and for \( \SP (n)=b(n) \) and
any \( \SP \TI  \)-constructible \( \TI (n) \) such that \( \SP \lessthan \TI \lessthan 2^{\SP } \),
\( \RST {\SP }{\TI }\neq \ST {\SP }{\TI } \).
\end{theorem}

\begin{proof}
\emph{(Sketch following proof of theorem~\ref{thm:revsep}.)} For input \( I \) of length \( n=b2^{b} \),
define result bit \( r(I) \)\label{p:result} to be the first bit in the \( b \)-bit
string given by \[
\underbrace{I[I[\ldots I[}_{\lfloor \TI /\SP \rfloor }{\texttt {0}}^{b}]\ldots ]].\]
 Let language \( L=\{I:r(I)={\texttt {1}}\} \)\label{p:Lnonrel}. \( L\in \ST {\SP }{\TI } \)
because an irreversible machine can simply follow the chain of \( \lfloor \TI /\SP \rfloor  \)
pointers from address \( {\texttt {0}}^{b} \), using space \( \Atmost (\SP ) \)
(not counting the input) and time \( \Atmost (\TI ) \).

Assume there is a reversible machine \( M \)\label{p:Mnonrel} that decides
\( L \) in \( c+c\SP  \)\label{p:cnonrel} space and \( c+c\TI  \) time for
some \( c \). Let \( b \) be sufficiently large for the proof below to work.
Let \( s \)\label{p:snonrel} be a certain description system to be defined.
Let \( t=\lfloor \TI /\SP \rfloor . \) Let \( x \) be a length-\( t\SP  \)
string incompressible in \( s \). Let \( w_{1}\ldots w_{t}=x \)\label{p:wi-nonrel}
where all \( w_{i} \)\label{p:i-nodeindex} are size \( b \). Restrict \( s \)
so that all the words \( w_{i} \) must be different from each other and from
\( 0^{b} \). Let \( I \) be an input of length \( n=b2^{b} \) such that \( I[{\texttt {0}}^{b}]=w_{1} \),
and \( I[w_{i}]=w_{i+1} \) for \( 1\leq i<t \), and \( I[a]={\texttt {0}}^{b} \)
for every other address \( a \). \( M \) must at some time access \( I[w_{t-1}] \)
because otherwise we could change the first bit of \( I[w_{t-1}] \) to be the
opposite of whatever \( M \)'s answer is, and \( M \) would give the wrong
answer. Assign a set of pebbled nodes to each configuration of \( M \)'s execution
on input \( I \) like in the oracle proof, except that this time, input access
operations take the place of oracle calls. Show, as in lemma~\ref{lemma:pebspace},
that the size of any of these configurations is at least \( \frac{1}{4}p\SP  \)
where \( p \) is the number of pebbled nodes, by defining \( s \) to allow
descriptions that are interpreted by simulating \( M \) forwards or backwards
and reading pebbled nodes from the input access tape. As before, the machine
must therefore take space \( \Atleast (\SP \log (\TI /\SP )) \), which for
sufficiently large \( n \) contradicts our assumption that the space is bounded
by \( c+c\SP  \). Thus \( L\notin \RST {\SP }{\TI } \).
\end{proof}

\begin{corollary}
\emph{Non-relativized lower bound on space for linear-time reversible simulations.}
For \( \SP =b(n) \), computable \( \SP' \lessthan \SP \log (\TI /\SP ) \),
and \( \SP \TI  \)-constructible \( \TI (n) \) such that \( \SP \lessthan \TI \lessthan 2^{\SP } \),
\( \RST {\SP' }{\TI }\nsupseteq \ST {\SP }{\TI } \). 
\end{corollary}

\begin{proof}
\emph{(Sketch.)} As in corollary~\ref{cor:lowerbound}, but with theorem~\ref{thm:nonrelsep}.
\end{proof}

Such a \( \TI  \) exists because \( b \) can be found in space and time \( \Atmost (\log b) \)
using the {}``get input size{}'' operation, after which \( \TI =b^{2} \),
for example, can be found in space \( \Atmost (\log b) \) and time \( \Atmost (\log ^{2}b) \).

\begin{corollary}
Thus, any reversible machine that simulates irreversible ones without asymptotic
slowdown takes \( \Atleast (\SP \log (\TI /\SP )) \) space in some cases, given
the type of input model presented in this section. 
\end{corollary}
Again, we emphasize that
this particular lower bound is probably not tight.

\noheading

We should also note that this particular non-relativized result is not very
compelling, because with a space bound that is much less than the input size,
the space usage is unlikely to reflect a dominant component of system cost for
real-world applications.

\section{Decompression algorithm}\label{s:alg}

It is probably not obvious to the reader that the algorithm that we
briefly mentioned in the proof of lemma~\ref{lemma:pebspace} in
\S\ref{s:lemma} can be made to work properly. In this section we give
the complete algorithm and explain why it works.

The algorithm, shown in figure \ref{a:decomp}, essentially just simulates \( M_{i} \)'s
operation in direction \( D \) starting from configuration \( C_{\tau } \),
and reads the identifiers of the pebbled nodes off of \( M_{i} \)'s simulated
oracle tape. The bulk of the algorithm is in the details showing how to simulate
all oracle queries correctly.

There is a small subtlety in the fact that this algorithm has, built into it,
some of the values of \( f \) that are defined by the oracle. Yet the algorithm
is part of the definition of our description system \( s_{i} \), which is used
to pick \( x \) and define the \( f(q_{j}) \) values. This would be a circularity
that might prevent the oracle from being well-defined, if not for the fact that
the portion of \( f \) that is built in, that is, \( f(b) \) for \( |b|<\SP  \),
is disjoint from the portion of \( f \) that depends on this algorithm, that
is, only values of \( f(b) \) for \( |b|\geq \SP (n_{i}) \). Thus there is
no circularity.

The \( f() \) values for the entire infinite oracle can be enumerated by enumerating
all values of \( i \) in sequence, and for each one, computing the appropriate
values of \( M_{i} \) and \( c_{i} \), and choosing an \( n_{i} \) that satisfies
all the explicit and implicit lower bounds on \( n \) that we mentioned above.
Then, \( n_{i} \) is used in the above algorithm to allow us to define \( s_{i} \)
and choose the appropriate \( x \), which determines \( f(b) \) for all \( b \)
where \( |b|=\SP (n_{i}) \); these values of \( f \) can then be added to
the table for use in the algorithm later when running on higher values of \( i \).

We now explain why the simulation carried out by the (oracle-less) decompression
algorithm imitates the real oracle-calling program exactly. When we come to
an oracle query operation where the queried bit-string(s) do not appear in our
\( q[j] \) array and do not have a matching \( \Delta \tau _{j} \), then we
know the bit-string(s) must not correspond to a real node in \( q_{1},\ldots ,q_{t} \),
because if they did, then either they were not pebbled due to queries in direction
\( D \), in which case they would have been in the description \( d \) and
would have been present in the initial \( q \) array, or else the first query
that involved them must have been before the current one (or else some \( \Delta \tau _{j} \)
would match), in which case they would have been added to the \( q \) array
earlier.

Moreover, when we get to a single query \( q_{j} \), we know we can look up
\( q_{j+1} \) to answer the query, because it must already have been stored.
Either \( q_{j+1} \) was not pebbled in direction \( D \) in which case it
was stored originally, or it was pebbled in direction \( D \) in which case
the first query involving it must have been before this one, since this query
is not of the type that would have caused the node to be pebbled in direction
\( D \). In either case we will already have a value in array entry \( q[j+1] \).

Given any description \( d \) derived from the execution history of a
real \( M_{i} \), the simulation will eventually find values for all
nodes, since either they were given initially or they are found
eventually as we simulate.  Thus the algorithm prints \( x \), as
required for the proof of lemma~\ref{lemma:pebspace}.

\begin{figure}
\begin{algorithm}[DECOMPRESS($d$)]
Given description \(d\) as described on p.\ \pageref{p:trickyone},
Let {\(q[1]\ldots q[t]\)} be a table of node values, initially all NULL.
Initialize all \(q[j]\)'s not pebbled in direction \(D\), as specified by description \(d\).
Simulate \(M_{i}\) in direction \(D\) starting from configuration \(C_{\tau}\), as follows:
\     To simulate a single operation of \(M_{i}\):
\          If it's a non-query operation, then
\               simulate it straightforwardly, and proceed.
\          Otherwise, it's an oracle query; examine the oracle tape.
\          If the tape is not of the form \(b\) or \(b\){\texttt{\#}}\(c\) for bit-strings \(b,c\), where \(|b|=|c|\),
\               do nothing for this operation.
\          Else, if \(|b|<\SP\), then look up \(f(b)\) in a computable table,
\               set the oracle tape appropriately, and proceed.
\          Else, if \(|b|>\SP\),then  do nothing for this operation.
\          Else, if the oracle tape is of the form \(b\), then
\               If the current step count matches some \(\Delta\tau_{j}\) in direction $D$, 
\                    then set \(q[j]=b\).
\               If \(b = q[j]\) for some \(j<t\),
\                    then set the oracle tape to \(b\){\texttt{\#}}\(q[j+1]\),
\               else, do nothing for this operation.
\          Else, if the oracle tape is of the form \(b\){\texttt{\#}}\(c\), then
\               For each \(\Delta\tau_{j}\) in direction $D$ matching the current step count,
\                    set \(q[j]\) to \(b\) or \(c\) depending on tag \(k_{j}\).
\               If \(b=q[j]\) and \(c=q[j+1]\) for some \(j\), set oracle tape to \(b\),
\               else do nothing for this operation.
\     Increment count of the number of steps simulated.
\     Continue simulating steps of $M_i$ until step count exceeds largest \(\Delta\tau_{j}\).
Print all \(q[j]\)'s.
\end{algorithm}

\caption{\label{a:decomp}Algorithm to print the incompressible chain of nodes \protect\( x\protect \)
via simulation of the reversible machine \protect\( M_{i}\protect \). }
\end{figure}

\ed{Possible additions to make the proof clearer: (1) Give pseudocode for the entire algorithm that computes \(A()\).}

\section{Beyond this proof}

In light of the work above, an obviously desirable next step would be to show
that \( \RST {\SP }{\TI }\neq \ST {\SP }{\TI } \) (and demonstrate corresponding
tight lower bounds) for a larger class of space-time functions \( \SP ,\TI  \)
in a reasonable serial model of computation \emph{without} an oracle or a black-box
input. A similar problem of following a chain of nodes may still be useful for
this. But when there is no oracle, and when the time bound is larger than the
input length \( \TI \morethan n \), there is no opportunity to specify an incompressible
chain of nodes to follow. Instead, the function \( f \) mapping nodes to their
successors must be provided by some actual computation that is specified by
the relatively short input. It may be helpful in such a proof if \( f \) is
non-invertible, or is a one-way invertible function, whose inverse might be
hard to compute. But \( f \) will still have some structure in general, and
so it may be very difficult to prove that there are no shortcuts that might
allow the result of repeated applications of \( f \) to be computed reversibly
using little time or space.

\section{Conclusion}

Although the above results are inconclusive with respect to their real-world
implications, it seems likely that reversible algorithms in the real world will
indeed in many cases require algorithmic space-time costs that exceed those
of traditional computations, by factors that are at least logarithmic and more
likely small polynomials in the cost of the original computation. 

However, this is not to say that reversible computing will never be useful.
For contexts where the cost of energy is high compared to the cost of computation,
or where the computation would benefit from a 3-D parallel architecture which
would tend be difficult to cool effectively, we have shown elsewhere that the
overall cost per performance of a partially-reversible solution may be lower
than that of a traditional irreversible design, despite the higher algorithmic
costs \cite{Frank-Knight-98}. Also, someday we might carry out \emph{quantum}
computations which would be demonstrably much \emph{more} efficient than traditional
computation on some problems, despite their reversibility \cite{Shor-95}. So,
the exact magnitude of the algorithmic cost of reversibility is still important,
because it affects the location of the optimal tradeoff points for a reversible
design, within those contexts where it \emph{is}\/ useful.

We believe that the most fruitful direction for future work in reversible computing
theory at this point is to optimize the parameters of Bennett's algorithm in
a way that minimizes the hardware cost per unit performance of parallel reversible
architectures, as a function of whatever upper bounds on power dissipation per
unit performance may arise from the requirements presented by particular application
contexts. Such analyses should take into account the asymptotic behavior of
realistic physical implementations of reversible computing; for example, there
is an additional asymptotic slowdown factor not accounted for in the present
paper which is required for the quasi-adiabatic (\ie, asymptotically reversible)
physical operation of real logic devices. Additionally, in order for the analytical
model to be useful for estimating the feasibility of real-world computer designs,
the model would also need to incorporate the specific constant factors and limits
on reversibility that would be incurred in a specific, feasible real-world reversible
technology. 

In the years since the first manuscripts of this paper were written
and circulated (circa 1997), we have been developing elements of some
practical reversible hardware technologies (\cf\
\cite{Frank-etal-98a,Vieri-etal-98,Vieri-etal-98b,Ammer-etal-99}) and
carrying out the accompanying tradeoff analysis.  The results of the
most recent (and still unpublished) work will be announced in future
reports to be presented to the computer science \& engineering
community.







\begin{acknowledgment}
Thanks are due to Michael Sipser of MIT for suggesting the use of the
incompressibility method, and to Alain Tapp and Pierre McKenzie of the
University of Montreal for helpful discussions and detailed feedback
on early drafts of this paper.
\end{acknowledgment}

\newpage

\nocite{PhysComp-92}
\bibliographystyle{unsrt}
\bibliography{shortstr,qcstr,refs,mpf,rcthy,thermcomp,moore,qc,mit,entropy,adia,circuits,physcomp,nano,tangent,xrefs}


\end{article}
\end{document}

%% file: macros.tex
\newcommand{\captionpars}{\addtolength{\parindent}{\outerparindent}}

\newcommand{\noheading}{\vspace{1.5ex}\noindent}

\font\tencmmib=cmmib10  \skewchar\tencmmib='177
\font\sevencmmib=cmmib7 \skewchar\sevencmmib='177
\font\fivecmmib=cmmib5  \skewchar\fivecmmib='177
\newfam\cmmibfam
\textfont\cmmibfam=\tencmmib \scriptfont\cmmibfam=\sevencmmib
 \scriptscriptfont\cmmibfam=\fivecmmib
\font\tencmbsy=cmbsy10  \skewchar\tencmbsy='60
\font\sevencmbsy=cmbsy7 \skewchar\sevencmbsy='60
\font\fivecmbsy=cmbsy5  \skewchar\fivecmbsy='60
\newfam\cmbsyfam
\textfont\cmbsyfam=\tencmbsy \scriptfont\cmbsyfam=\sevencmbsy
 \scriptscriptfont\cmbsyfam=\fivecmbsy

\mathchardef\biC="0\hexnumber@\cmmibfam43

\newcommand{\defn}{\textit}	   
\newcommand{\foreign}{\textit}     
\newcommand{\latin}{\foreign}	   

\newcommand{\ed}[1]{}				

\newcommand{\Ie}  {\latin{I.e.}}   
\newcommand{\ie}  {\latin{i.e.}}   
\newcommand{\etal}{\latin{et al.}} 
\newcommand{\Eg}  {\latin{E.g.}}   
\newcommand{\eg}  {\latin{e.g.}}   
\newcommand{\cf}  {\latin{cf.}}	   

\newenvironment{minitab}
  {\begin{minipage}{\textwidth}
   \begin{tabbing}}
  {\end{tabbing}
   \end{minipage}}

\newcommand{\beforetab}{\addtolength{\parindent}{\outerparindent}}
\newcommand{\aftertab}{\hrule}
\newlength{\outerparindent}
\input{mathmacros.tex}

%% file: mathmacros.tex
\newcommand{\rmn}[1]{{\rm #1}}  

\newcommand{\litsym}  [1] {{\tt #1}}            
\newcommand{\cxclass} [1] {{\bf #1}}            
\newcommand{\unit}        {\rmn}	        
\newcommand{\units}   [1] {\enm{\;\mathrm{#1}}} 
\newcommand{\cssym}   [1] {\enm{\mathsf{#1}}}   
\newcommand{\field}   [1] {\enm{\mathbb{#1}}}   

\newcommand{\e}      [2] {#1\!\times\!10^{#2}}                
\newcommand{\tg}         {_\rmn}                              
\newcommand{\enm}        {\ensuremath}                        
\newcommand{\funcdr} [3] {\enm{#1\colon#2\rightarrow#3}}      
\newcommand{\aprxeq}     {\enm{\mathrel{\cong}}}              
\newcommand{\about}      {\enm{{\sim}}}			      
\newcommand{\qeq}     {\stackrel{?}{=}}              
\newcommand{\zerobit} {\litsym{0}}                   
\newcommand{\onebit}  {\litsym{1}}                   
\newcommand{\hash}    {\enm{\text{\litsym{\#}}}}     
\newcommand{\bitset}  {\{\zerobit,\onebit\}}         

\newcommand{\N}	     {\field{N}} 

\renewcommand{\P}     {\cxclass{P}}       
\newcommand{\NP}      {\cxclass{NP}}      
\newcommand{\RTISP}   {\cxclass{RTISP}}	  
\newcommand{\TISP}    {\cxclass{TISP}}    
\newcommand{\RST} [2] {\enm{\RTISP(#2,#1)}} 
\newcommand{\ST}  [2] {\enm{\TISP(#2,#1)}}  

\newcommand{\Atmost}   {\enm{\mathcal{O}}}       
\newcommand{\Atleast}  {\enm{\boldsymbol\Omega}} 
\newcommand{\Exactly}  {\enm{\boldsymbol\Theta}} 
\newcommand{\Lessthan} {\enm{\mathbf{o}}}        
\newcommand{\Morethan} {\enm{\boldsymbol\omega}} 

\newcommand{\atmost}   {\precsim} 
\newcommand{\atleast}  {\succsim} 
\newcommand{\exactly}  {\asymp}     
\newcommand{\lessthan} {\prec}    
\newcommand{\morethan} {\succ}    

\newcommand{\gram}   {\unit{g}}
\newcommand{\Joule}  {\unit{J}}
\newcommand{\Kelvin} {\unit{K}}
\newcommand{\Newton} {\unit{N}}

\newcommand{\meter}  {\unit{m}}
\newcommand{\second} {\unit{s}}          

\newcommand{\cm}     {\unit{cm}}         
\newcommand{\kg}     {\unit{kg}}         

\newcommand{\kB}  {\enm{k\tg{B}}}       

\newcommand{\tapespace} {\cal C}		  
\newcommand{\klass}	{\enm{\mathfrak{C}}}	  
\newcommand{\reduc}	{\enm{F}}	          
\newcommand{\comboi}	{\enm{i}}		  
\newcommand{\kT}	{\enm{\kB\temperature}}	  
\newcommand{\machine}   {\enm{M}}           
\newcommand{\nsubi}	{\enm{n_\comboi}}   
\newcommand{\probsize}  {\enm{n\tg{in}}}    
\newcommand{\nin}       {\probsize}         
\newcommand{\prem}	{\enm{P}}	    
\newcommand{\Maxmem}      {\cssym{S}}	      
\newcommand{\SP}          {\Maxmem}           
\newcommand{\descsys}	  {\enm{s}}	      
\newcommand{\si}	  {\enm{s_\comboi}}   
\newcommand{\temperature} {\enm{T}}	      
\newcommand{\Nticks}      {\cssym{T}}	      
\newcommand{\TI}          {\Nticks}           
							      
